\documentclass[11pt]{amsart}
\usepackage{amssymb,mathrsfs,graphicx,extpfeil}
\usepackage{epsfig}
\usepackage{indentfirst, latexsym, amssymb, enumerate,amsmath,graphicx}
\usepackage{float}
\usepackage{colortbl}
\usepackage{epsfig,subfigure}
\usepackage[pagewise]{lineno}
\usepackage{mathtools}
\usepackage{xcolor}
\usepackage{enumitem}
\usepackage{graphicx}
\usepackage{mathtools}
\usepackage{verbatim} 
\usepackage{bm} 
\PassOptionsToPackage{hyphens}{url} 
	\usepackage[colorlinks=true]{hyperref}
	\usepackage{url}

\usepackage{amsmath,amssymb, amsfonts,extarrows,dsfont}
\usepackage{xcolor}
\usepackage{enumitem}
\usepackage{graphicx}

\allowdisplaybreaks

\newcommand{\Var}[1]{\ensuremath{\mathrm{Var}\left(#1\right)}}
\newcommand{\<}{\kern-2pt\left\langle}
\renewcommand{\>}{\right\rangle\kern-2pt}

\renewcommand{\d}[1]{\ensuremath{\mathrm{d} #1}}

\newtheorem{thm}{Theorem}
\newtheorem{defn}{Definition}
\newtheorem{lem}{Lemma}

\newtheorem{rem}{Remark}

\newcommand{\tos}{\xlongrightarrow[N\to\infty]{\text{in prob}}}

\begin{document}

\title[Mean-field equation for a model of quorum-sensing microbial populations]{Mean-field equation for a stochastic many-particle model of quorum-sensing microbial populations
}



\author[E. Frey]{Erwin Frey}
\address[Erwin Frey]{\newline Arnold Sommerfeld Center for Theoretical Physics and Center for NanoScience, Department of Physics, Ludwig-Maximilians-Universit\"at M\"unchen, Theresienstrasse 37, 80333 M\"unchen, Germany \\}
\email{frey@lmu.de}

\author[J. Knebel]{Johannes Knebel}
\address[Johannes Knebel]{\newline Arnold Sommerfeld Center for Theoretical Physics and Center for NanoScience, Department of Physics, Ludwig-Maximilians-Universit\"at M\"unchen, Theresienstrasse 37, 80333 M\"unchen, Germany \\}
\email{johannes.knebel@physik.lmu.de}           

\author[P. Pickl]{Peter Pickl}
\address[Peter Pickl]{\newline Department of Mathematics, Ludwig-Maximilians-Universit\"at M\"unchen, Theresienstrasse 39, 80333 M\"unchen, Germany \\}
\email{pickl@math.lmu.de}




\maketitle

\begin{abstract} 
We prove a mean-field equation for the dynamics of quorum-sensing microbial populations. 
In the stochastic many-particle process, individuals of a population produce public good molecules to different degrees. Individual production is metabolically costly such that non-producers replicate faster than producers. In addition, individuals sense the average production level in the well-mixed population and adjust their production in response (``quorum sensing'').
Here we prove that the temporal evolution of such quorum-sensing populations converges to a macroscopic mean-field equation for increasing population sizes. To prove convergence, we introduce an auxiliary stochastic mean-field process that mimics the dynamics of the mean-field equation and that samples independently the individual’s production degrees between consecutive update steps. This way, the law of large numbers is separated from the propagation of errors due to correlations. 
Our developed method of an auxiliary stochastic mean-field process may help to prove mean-field equations for other stochastic many-particle processes.
%
\end{abstract}
\tableofcontents

\section{Introduction}
\label{sec:intro}

\paragraph{Background: Effective descriptions of stochastic many-particle models in biological physics.}
The dynamics of biological systems are often modeled in terms of stochastic many-particle processes. For example, the temporal evolution of a microbial population can be suitably described in terms of a stochastic birth-death process to model the competition between species and their long-time evolution~\cite{Nowak2004,Blythe2007,Weber2017}. 
The dynamics of assemblies of epithelial cells in a tissue in terms of cellular Potts models~\cite{Graner1992,Sepulveda2013,Segerer2015}, the spatial organization of cellular components such as actin filaments in terms of active matter models~\cite{Vicsek1995,Schaller2010,Marchetti2013}, and intra-cellular transport of molecular cargo on such filaments in terms of the totally asymmetric simple exclusion process~\cite{Spitzer1970,Krug1991,Derrida1992,Klumpp2003,Parmeggiani2003,Blythe2007b,Chou2011} are further examples for the successful application of stochastic processes in biological physics.

Stochastic many-particle processes are typically formulated for certain  microscopic degrees of freedom of the biological system under consideration.
To verify or falsify theoretical predictions with experimental observables that characterize the various phenomenological phases of the respective biological system, effective descriptions of the underlying stochastic many-particle processes are needed. 
Such effective descriptions are also important to both identify the parameters and quantify their regimes that may be promising to study in an actual experiment.\\
As an example, consider the gene expression of individual cells in a microbial population that grows in a well-stirred flask. 
From an experimental point of view, the cellular production of a specific protein  can be linked to the expression of a fluorescent reporter protein such as the Green Fluorescent Protein. 
However, it is often not possible to resolve the engineered gene expression at the single-cell level in the desired environment. 
In other words, it may not be possible to observe which cell produces the fluorescent protein to which degree.
Instead, the question of \textit{how many} cells (as opposed to \textit{which} single cell) produce a certain degree of fluorescent proteins may be experimentally accessible, for example, through fluorescence-activated cell sorting (FACS). \\
In mathematical terms, only the empirical one-particle density of the microbial population would be experimentally accessible in this example (\textit{how many cells produce to which extent}), but not the joint probability distribution of all individuals (\textit{which single cell produces to which extent}).
Despite this experimental restriction, the measurement of the one-particle density can be used to discriminate between different phases of the microbial population. 
For example, whether all individuals produce a protein to the same degree (homogeneous phase) or to different degrees (heterogeneous phase) can be discriminated by the shape of the one-particle density (unimodal or bimodal).
In essence, the dynamics of the one-particle density is an effective description at the population level for the stochastic many-particle process, which describes the dynamics of protein production at the single-cell level in the above example, and distinguishes between homogeneous and heterogeneous phases.

Effective descriptions of many-particle processes in biological physics are often guessed by making use of symmetry arguments, by assuming linear response, or are heuristically derived through kinetic theories~\cite{Edelstein1988,deGennes1995,Zwanzig2001,Murray2002,Marchetti2013,Kadar2011,Balian2007}. 
To validate and quantify the scope of validity of effective descriptions, to quantify the speed of convergence, and to show robustness of the modelling approach in the first place, it is necessary to prove the convergence with respect to the system size of the stochastic many-particle process to the effective description.
However, proving such effective descriptions is often cumbersome, which is a well-known problem in various fields of physics such as in statistical physics (Boltzmann equation~\cite{Lanford1975}, Vlasov equation~\cite{Neunzert1974,Braun1977}) and in biological physics (Keller-Segel equation~\cite{Cattiaux2016,Huang2017,Canizares2017}, Cucker-Smale~\cite{Carlen2013,Figalli2017}), and in the field of social sciences such as for pedestrian flows~\cite{Yin2017} and opinion dynamics~\cite{BenNaim2003,Lorenz2007}.

\paragraph{Main result of this work: Convergence to mean-field for large $N$ in the quorum-sensing model.}
In this manuscript, we prove the validity of an effective description for the so-called \textit{quorum-sensing model}, which is a stochastic many-particle process for the temporal evolution of a quorum-sensing microbial population of $N$ individuals~\cite{Bauer2017}. 

In the quorum-sensing model, each individual produces public good molecules to different degrees and secretes those into the well-mixed environment; see illustration in Figure~\ref{fig:model}(B). Individual production is metabolically costly, such that non-producers replicate faster than producers. In addition, individuals sense the average production level in the well-mixed population and adjust their production in response (known as ``quorum sensing''). 
As was recently shown, depending upon both the response rate and the growth rate differences, the microbial population can evolve in time into a homogeneous phase, in which all bacteria produce to the same degree, or into a heterogeneous phase, in which the population splits into two subpopulations with different production degrees~\cite{Bauer2017}. 
These theoretical findings might explain recent experiments, and challenge currently accepted views on phenotypic heterogeneity in quorum-sensing microbial populations.\\
Previously, observations of numerical simulations of the quorum-sensing model were explained by the analysis of an effective mean-field equation~\eqref{eq:mean-field}~\cite{Bauer2017}.
This mean-field equation describes the temporal evolution of the distribution of  production degrees in the population.
The derivation of this effective description included two steps: 
(i) the derivation of the temporal evolution equation of the reduced one-particle density $\rho^{(1)}$ for the production degrees in the population, which was obtained from the continuous-time Markov process describing the dynamics of the single-cell production degrees;
(ii) the assumption of a mean-field density $\rho$ and the heuristic guess of the mean-field equation~\eqref{eq:mean-field} that governs the temporal evolution of $\rho$. 

Here, we prove that and quantify how the stochastic process of the quorum-sensing model converges to the mean-field equation~\eqref{eq:mean-field} as the population size grows to infinity ($N\to \infty$); see Figure~\ref{fig:outline_proof}. 
More precisely, we establish that for any finite time $t$, the empirical density of the underlying microscopic process $\rho^{(1)}_N$ converges in probability to the mean-field density $\rho$ of the effective dynamics as $N\to \infty$ if initial correlations are not too strong, see Theorem~\ref{thm:main}.
In other words, our proof yields concrete error bounds for the quality and the speed of convergence of the stochastic process to the mean-field equation. 
These error bounds depend upon both the population size and the initial correlations.
To show convergence, we introduced an auxiliary stochastic mean-field process that mimics the temporal evolution of the mean-field equation~\eqref{eq:mean-field} and updates the individuals' production degrees in an independent manner; see Figure~\ref{fig:macroscopic}.
This way, the law of large numbers is separated from the propagation of errors that build up due to correlations between the production degrees of the individuals, and enables to prove convergence towards mean-field as $N\to \infty$.

\paragraph{Significance of our results in the biological context.}
The proof presented in this manuscript also shows that the convergence to the mean-field equation~\eqref{eq:mean-field} is robust against changes of microscopic details in the definition of the quorum-sensing model.
In particular, convergence to mean-field does not depend upon the specific choice of the fitness function and the implemented way by which cells sense and respond to the average production level of the public good in the environment. 
Previously, it was already shown that the occurrence of heterogeneous phases of the population depends only upon the qualitative behavior of both the fitness function and the response function (that is, their respective fixed point structure)~\cite{Bauer2017}.
Together with the results of this work, these robustness properties are important for the applicability of the quorum-sensing model to a biological, \textit{in vivo} or \textit{in vitro}, experiment, for which access to the qualitative form of both fitness and response might be possible~\cite{Ruparell2016,Diggle2007,He2003,Williams2008,Drees2014,Hense2015,Maire2015}.
In total, the convergence to mean-field presented in Theorem~\ref{thm:main} and its proof support the biological relevance of the quorum-sensing model for phenotypic heterogeneity in quorum-sensing microbial populations.

\paragraph{Significance of our applied method for other stochastic many-particle models.}
Furthermore, we expect that our developed method of an auxiliary stochastic mean-field process could also be helpful to prove the convergence of other stochastic many-particle processes towards their respective mean-field equations. 
Birth and death processes in discrete systems -- for example on networks -- are relevant in many applications describing physical, biological, or social systems. In some of these systems, a rigorous proof of the validity of an effective description could be possible using the techniques of this manuscript.   

The idea of using an auxiliary process for proving mean-field limits for many-body processes has already been successfully implemented in the continuous case, for example, in~\cite{Boers2015,Huang2017,Canizares2017}.
By combining these ideas with the technique of the present paper, it should be possible to generalize most of these results to the case including a random birth- and/or death process. 
For example, in reference~\cite{Huang2017,Canizares2017}, the Keller-Segel equation is derived from a microscopic model describing the continuous motion of interacting amoebas by comparing their  microscopic  dynamics to an auxiliary system, defined by the trajectories that follow the mean-field flow. 
However, discrete gain and loss terms naturally occur in the description of life cycles in a colony of amoebas, which can be described by adding a birth and/or death process to the microscopic model. We envision that it should be possible to obtain and prove a Keller-Segel-like equation including a gain and/or loss term as effective description by using an auxiliary stochastic mean-field process as presented in this manuscript.

\paragraph{Outline of the manuscript.}
This manuscript is organized as follows. 
In Section~\ref{sec:model}, we introduce the quorum-sensing model as a stochastic many-particle process to describe the temporal evolution of quorum-sensing microbial populations; Figure~\ref{fig:model}  illustrates the model set-up, both in the individual-based and population-based description. We provide a brief overview of previous numerical and analytical results of the quorum-sensing model, and explain why this model might have biological relevance in the context of phenotypic heterogeneity of autoinducer production in quorum-sensing microbial populations. 
The mean-field equation~\eqref{eq:mean-field} is introduced as an effective description of the temporal evolution of quorum-sensing microbial populations, which explains the observed phenotypic heterogeneity.
Our main result for the convergence of the empirical one-particle density (derived from the microscopic, stochastic many-particle process) towards the mean-field density (the effective  or macroscopic process) is presented in Theorem~\ref{thm:main} in Section~\ref{sec:main_result}.
Figure~\ref{fig:outline_proof} in Section~\ref{sec:main_result} summarizes the key steps of the proof. 
Essential to the idea of the proof is the introduction of an auxiliary stochastic mean-field process, which is explained in Section~\ref{sec:auxiliary_process} and illustrated in Figure~\ref{fig:macroscopic}. 
The details of the proof are explained in Section~\ref{sec:proof}.

\section{The quorum-sensing model: a stochastic many-particle process for the temporal evolution of quorum-sensing microbial populations}
\label{sec:model}

We now introduce the set-up of the quorum-sensing model in the individual-based description by closely following along the lines of reference~\cite{Bauer2017}.

\begin{figure}[htb!]
  \centering
  \includegraphics[width=0.99\textwidth]{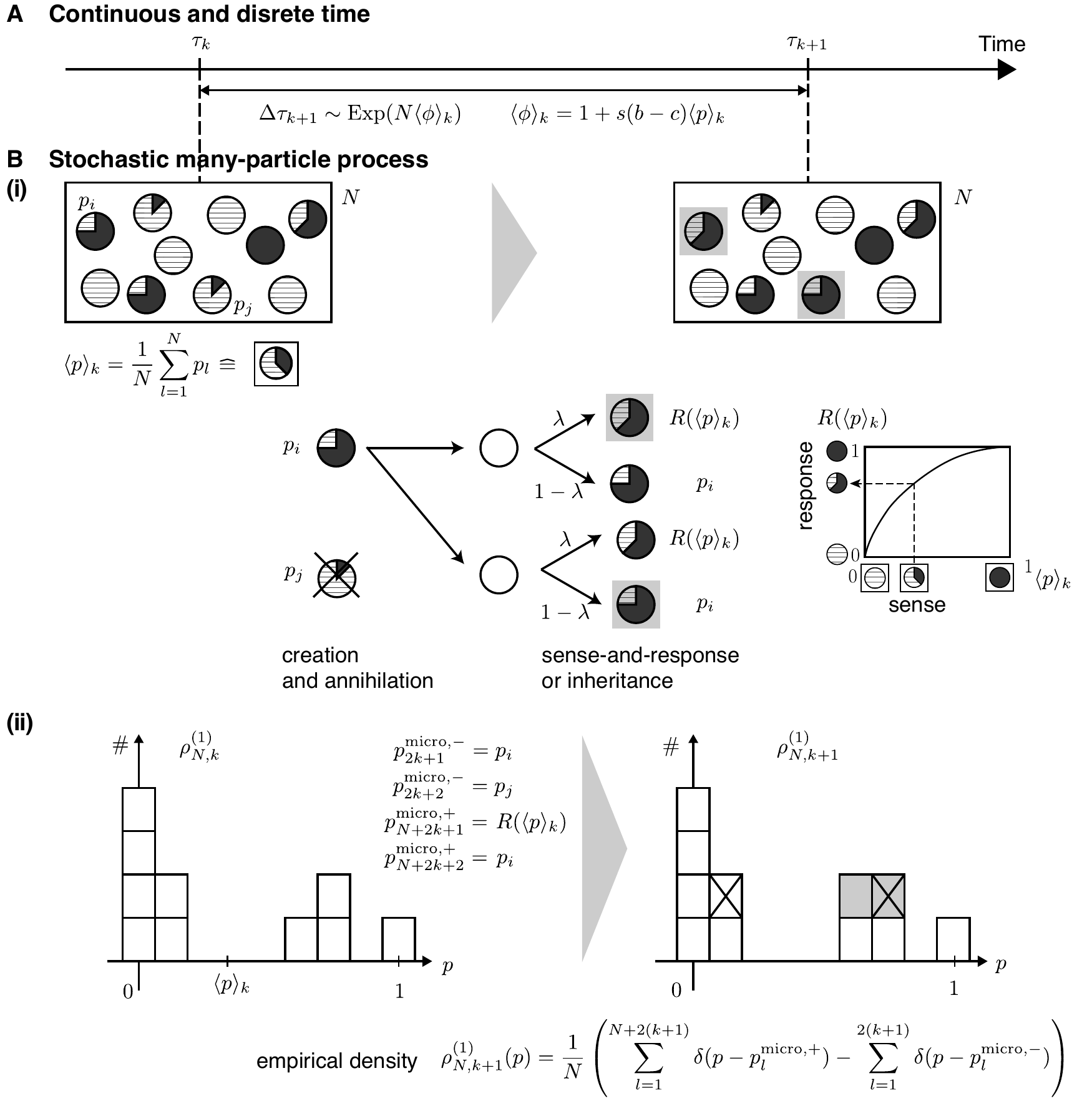}
  \caption{
Set-up of the quorum-sensing model as a stochastic many particle process.
(A) Update steps of the population occur continuously in time after an exponentially distributed waiting time $\Delta \tau_{k+1}$ that depends on the fitness of the population $N\langle\phi\rangle_k$. The discrete label $k$ is approximately a function of the continuous time~$t$; see Section~\ref{sec:microscopic_revisited} and Lemma~\ref{lem:time}. 
(B) Sketch of one update step in the quorum-sensing model. 
(i) Individual-based description. Individuals are depicted as disks and are characterized by their production degree $p_i\in [0,1]$ (indicated by fraction of dark gray filling; non-producers with $p_i=0$ are indicated by line pattern and reproduce fastest, full producers with $p_i=1$ are depicted in dark gray and reproduce slowest, see Equation~\eqref{eq:fitness}).
At an update step, individual $i$ (here $p_i=3/4$) divides into two offspring individuals, one of which replaces an individual $j$ that is randomly chosen with probability $1/(N-1)$.
Each offspring individual senses the average production level in the population, $\langle p\rangle_k$ at time $\tau_k$, and either adopts, with probability $\lambda$, the value $R(\langle p\rangle_k)$ of the response function as its production degree or, with probability $1- \lambda$, inherits the production degree from the ancestor. Here, offspring individual $i$ inherits and $j$ adapts (denoted by light gray background). 
(ii) Population-based description. A single realization of the stochastic many-particle process can be reformulated at the level of the empirical (one-particle) density $\rho^{(1)}_N$, that is, the histogram of production degrees; see Section~\ref{sec:microscopic_revisited}. At the $(k+1)^\text{th}$ update step, the histogram changes due to two annihilated ($p^{\text{micro},-}_{2k+1}, p^{\text{micro},-}_{2k+2}$) and two created ($p^{\text{micro},+}_{2k+1}, p^{\text{micro},+}_{2k+2}$) production degrees, whose probability distributions are defined such that they agree with the individual-based description; see Definition~\ref{defn:micro_process} and explanations in Section~\ref{sec:microscopic_revisited}.
}
  \label{fig:model}
\end{figure}
%


\paragraph{State of the population.}
The quorum-sensing model is defined for a well-mixed population of $N$ individuals.  
Each individual $i=1,\dots, N$ is characterized by its production degree $p_i\in[0,1]$. 
The value of $p_i$ denotes the extent to which an individual produces the public good. This public good is secreted into the well-mixed population and, thus, becomes equally shared amongst all individuals of the population.
The limiting case $p_i=0$ denotes a non-producer and $p_i=1$ denotes a full producer. 
The state of the population at a given time is characterized by the collection of the production degrees of all $N$ individuals $P=(p_1, \dots,p_N)$.

\paragraph{A stochastic birth-death process governs the temporal evolution of the state of the population.}
 The state of the population $P$ changes through a continuous-time Markov process; see Figure~\ref{fig:model} for an illustration. The state $P$ is updated through both reproduction of individuals and a sense-and-response mechanism, such that at most two individuals $i$ and $j\neq i$ change their production degree at one time; see Figure~\ref{fig:model}(B). In the following, these processes that change the state of the population are explained.
\paragraph{Reproduction.}
An individual~$i$ replicates randomly (referred to as a birth or creation event) after a time that is exponentially distributed with rate $\phi_i$; see Figure~\ref{fig:model}(A). 
This replication rate $\phi_i$ is also referred to as the individual's fitness and depends on the individual's production degree and on the whole population as follows. 
On the one hand, fitness  decreases due to the individual's metabolic costs of production quantified by the value $cp_i$ (with cost unit $c\geq0$). In other words, non-producers replicate faster than producers. 
On the other hand, fitness increases with the available public good in the well-mixed population. Because the secreted public good is shared equally between all individuals, it is assumed that the fitness increases by the value $b\langle p\rangle$ (with benefit unit $b\geq 0$); 
here $\langle p\rangle= 1/N\sum_{i=1}^N p_i$ denotes the average production level in the population at that instant in time. In other words, the more producers are present in a population, the faster individuals replicate in that population.

To be explicit, the individual's fitness could be chosen as a linear function of costs and benefits: 
\begin{align}\label{eq:fitness}
\phi_i(P) =  1+s(b\langle p \rangle-cp_i)\ ,
\end{align}
and may be generalized to any function $\phi_i(P) =\phi(p_i, \langle p\rangle)$ that is bounded in Lipschitz norm (see definition in Equation~\eqref{eq:BL_norm} further below).
In the chosen fitness function~\eqref{eq:fitness}, the influence of the balance between costs and benefits is scaled by the selection strength $0<s<1/c$ and is added to the background fitness 1. 
The average per capita fitness is given by $\langle \phi\rangle = 1+ s(b-c)\langle p\rangle$.

Whenever an individual splits into two offspring individuals, another individual from the population is randomly selected with probability $1/(N-1)$ to die (referred to as a death or annihilation event) such that the population size $N$ remains constant. 
For the choice of the fitness function~\eqref{eq:fitness}, the time unit $\Delta t = 1$ means that, given a population consisting solely of non-producers, each individual will have replicated once on average. 

%
\paragraph{Sense-and-response through quorum sensing.}
Furthermore, individuals may change their production degree via quorum sensing. 
For simplicity, we implemented sense-and-response as follows; see also Figure~\ref{fig:model}(B) for an illustration. 
At a reproduction event, both offspring individuals of ancestor $i$ sense the average production level $\langle p\rangle$ in the well-mixed population. 
With probability $\lambda \in [0,1]$, they independently adopt the value $R(\langle p\rangle) \in [0,1]$ as their production degree, whereas they inherit the ancestor's production degree with probability $1-\lambda$. 
We refer to the function $R(\langle p\rangle)$ as the response function (see Figure~\ref{fig:model}(B) for a sketch), which is the same for all individuals.
In this implementation of sense-and-response, the response probability can be thought of as a response rate measured in units of the reproduction rate.
We refer to the ability of an individual to sense and respond to a property of the whole population (such as the average production level $\langle p\rangle$) as quorum sensing.
Importantly for the quorum-sensing model, sense-and-response constitute a source of innovation in the space of production degrees because an individual may adopt a production degree that was not previously present in the population.

Central to the quorum-sensing model is the feature that individuals shape their environment (through the production and secretion of the public good) and respond to this self-shaped environment (by changing their individual production), in turn. 
In simple terms, every individual feels the average production of all other individuals and adjusts its individual production accordingly.
Thus, ecological and population dynamics are coupled in the quorum-sensing model. 
This coupling results in interesting collective dynamics that are  summarized in the following.

\paragraph{Phenomenology of the quorum-sensing model and biological significance.}
Previous numerical simulations and mathematical analysis of the quorum-sensing model showed that the coupling between ecological and population dynamics through quorum sensing may induce phenotypic heterogeneity in the production of public goods in microbial populations~\cite{Bauer2017}. 
These findings qualitatively explain recent experimental observations in microbial population dynamics~\cite{Garmyn2011,Grote2015} and challenge currently accepted views on the origin of phenotypic heterogeneity~\cite{Ackermann2015} in quorum-sensing microbial populations. 

More specifically, upon numerically simulating the stochastic process of the quorum-sensing model, one observes both homogeneous and heterogeneous (quasi-)stationary phases of the population depending on the chosen parameter values for the selection strength $s$ and response probability $\lambda$.
For a broad range of values for $s$ and $\lambda$, the population evolves into a state in which all individuals produce public goods to the same degree (referred to as homogeneous phase, $p_i = p_j$ for all $i$ and $j$). 
For fixed selection strength~$s$, such homogeneous phases are approached if $\lambda = 0$ or if $\lambda$ exceeds some threshold value. Homogeneous phases of the population are intuitively expected to occur at least for some parameter regimes because every individual in the quorum-sensing model senses the production of every other individual and responds accordingly to the average.
Notably, however, the coupling between ecological and population dynamics can also yield to a stable heterogeneous production of public good molecules in the population.
For intermediate values of the response probability $\lambda$, the population may split into two subpopulations: one with a low, and a second with a high production of public good molecules. 
This heterogeneity in the public good production is stable for many generations. 
At the same time, the overall production level~$\langle p\rangle$ is robustly self-regulated if individual production is up-regulated through the response function $R(\langle p\rangle)$.
Such heterogeneous phases arise for diverse initial states of the population.
The phase transitions between homogeneous and heterogeneous (quasi-)stationary phases occur if cellular response to the environment is absent ($\lambda = 0$) or too frequent (high values of~$\lambda$).
In total, if individuals sense and respond to their self-shaped environment, the population may not only respond as a homogeneous collective as is typically associated with quorum sensing, but may also become a robustly controlled collective of two different subpopulations~\cite{Bauer2017}.

These findings may have direct applications to microbial population dynamics for which so-called autoinducers can be understood as public goods molecules. Autoinducers are small signaling molecules that enable microbes to communicate with each other in terms of a chemical language. These autoinducers are secreted into the environment and are sensed by other microbes in the population, in turn. 
Upon responding to the sensed level of autoinducers in the environment, a coordinated gene expression of all cells of the population can be triggered. Such collective behavior of microbes is commonly referred to as ``quorum sensing'' and comprises, for example, the coordinated and collective expression of genes for virulence, biofilm formation, and bioluminescence. 
Recent experiments suggest that the production of autoinducers may vary between genetically identical cells in a population in that some cells of the population expressed autoinducer synthase genes during microbial growth, while others did not~\cite{Garmyn2011,Grote2015}. 
Such a phenomenon is referred to as phenotypic heterogeneity~\cite{Ackermann2015}.
The stable coexistence of different phenotypes in one population may serve the division of labor or act as a bet-hedging strategy and, thus, may be beneficial for the survival and resilience of a microbial species at long time scales.
However, the experimentally observed phenotypic heterogeneity in the autoinducer production is not expected to occur in well-mixed populations if currently favored threshold models for quorum-sensing response are adopted~\cite{Bauer2017}.
The phenomenology of the quorum-sensing model shows that a microbial population can, indeed, control phenotypic heterogeneity of autoinducer production and, concomitantly, tightly adjust the average production level in the population to trigger quorum-sensing functions such as virulence.
In other words, the quorum-sensing model might be relevant to explain how phenotypic heterogeneity in the production of autoinducers is established in quorum-sensing microbial populations.

\paragraph{Macroscopic (population-based) description and effective temporal evolution: mean-field equation~\eqref{eq:mean-field}.}
To describe the numerically observed quasi-stationary, heterogeneous phases of the population, we previously derived heuristically a macroscopic mean-field equation from the microscopic stochastic many-particle process. All details of both the heuristic derivation and the mathematical analysis can be found in the supplement of reference~\cite{Bauer2017}\footnote{Please note that the \textit{average} one-particle density was denoted as $\rho_N^{(1)}$ in reference~\cite{Bauer2017}, whereas it denotes the \textit{empirical} density in this manuscript; see Equation~\eqref{eq:empirical_one-particle_density_individual_based}.}. 
The central steps are outlined in the following to motivate the proof for the convergence of mean-field.

On a macroscopic (that is, population-based)  description level, the population is suitably characterized by the so-called reduced one-particle density, 
\begin{align}\label{eq:one-particle_density_individual_based}
\rho^{(1)}(p,t) = 1/N \left\langle\sum_{i=1}^N\delta(p-p_i)\right\rangle_{\mathrm{P}(P,t)} \ ,
\end{align}
in the spirit of a kinetic theory~\cite{Kadar2011,Spohn1991}, in which $\rho^{(1)}$ denotes the probability distribution of finding \textit{any} individual at a specified production degree $p$ at time $t$; see Figure~\ref{fig:model}(B)(ii). 
We also refer to $\rho^{(1)}$ as the \textit{average density} of the microscopic process as opposed to the \textit{empirical density}:
\begin{align}\label{eq:empirical_one-particle_density_individual_based}
\rho_N^{(1)}(p,t) = 1/N\sum_{i=1}^N\delta(p-p_i)\ ,\end{align}
which is the histogram of production degrees of a single realization of the stochastic process at time $t$; see also Definition~\ref{defn:empdens} for a different formulation.
The average in the definition of the average density is taken over the joint $N$-particle probability distribution $\mathrm{P}(P,t)$. In other words, the value $\mathrm{P}(P,t)\mathrm{d}{p_1}\dots\mathrm{d}{p_N}$ denotes the joint probability of finding the first individual with a production degree in the interval $[p_1, p_1+\mathrm{d}{p_1}]$, the second individual with a production degree in the interval $[p_2, p_2+\mathrm{d}{p_2}]$, and so on at time $t$.
The temporal evolution of $\mathrm{P}(P,t)$ is governed by a master equation for the stochastic many-particle process~\cite{Gardiner,VanKampen2007,Weber2017}, which follows from the definition of the quorum-sensing model and tracks the correlated microscopic dynamics of the production degrees of all $N$ individuals. 
The temporal evolution of $\rho^{(1)}$ is derived from the master equation for $\mathrm{P}$; see~\cite{Bauer2017} for details. 
The average density $\rho^{(1)}$ may be approximated by the mean-field density $\rho$ if one naively assumes that correlations are negligible. Under this mean-field assumption, $\rho$ evolves according to the mean-field equation:
\begin{equation}\label{eq:mean-field}
\partial_t \rho(p, t) = 
2\lambda \langle \phi\rangle_{\rho_t}\big(\delta(p-R(\langle p\rangle_{\rho_t}))-\rho(p,t)\big)+ (1-2\lambda)\big(\phi(p, \langle p\rangle_{\rho_t})-\langle \phi\rangle_{\rho_t}\big)\rho(p,t)\ .
\end{equation}
Here, we abbreviated $\langle \phi\rangle_{\rho_t} = \int_0^1\d p\ \phi(p) \rho(p,t)$, and analogously for $\langle p\rangle_{\rho_t}$. 
The mean-field equation~(\ref{eq:mean-field}) conserves normalization of $\rho$, that is, $\int\d p\ \partial_t \rho(p,t) = 0$.
Note that the mean-field equation~\eqref{eq:mean-field} is to be understood in distributional sense, that is, it needs to be integrated over observables (for example, suitable test functions $g: [0,1]\to \mathbb{R}$ with bounded Lipschitz norm~\eqref{eq:BL_norm}) and $\rho$ is interpreted as a linear functional on the space of these observables. 
Two terms contribute to the mean-field equation~(\ref{eq:mean-field}) and determine  how the distribution of production degrees in the population evolves in time: the sense-and-response term with prefactor $2\lambda$, and the replicator term with prefactor $1-2\lambda$. 
When quorum sensing is absent ($\lambda = 0$), the sense-and-response term vanishes and Equation~(\ref{eq:mean-field}) reduces to the well-known replicator equation of the continuous Prisoner's dilemma~\cite{Bomze1990,Oechssler2001,Hofbauer2003,Cressman2005,McGill2007}. 
In general, the replicator term determines how  probability weight at production degree $p$ changes if the fitness $\phi(p)$ is different from the mean fitness in the population $\langle \phi\rangle_{\rho_t}$. The sign of the contribution of the replicator term changes when exactly one of the two offspring individuals adapts on average ($\lambda = 1/2$). 
The sense-and-response term, on the other hand, encodes the ecological feedback by which individuals sense the average $\langle p\rangle_{\rho_t}$ and adopt the production degree $R(\langle p\rangle_{\rho_t})$ in response. The change in $\rho$ at a certain production degree is determined by the difference between the current state $\rho$ and the state in which all individuals have this production degree $R(\langle p\rangle_{\rho_t})$. 

The analysis of the mean-field equation~\eqref{eq:mean-field} explains both homogeneous and heterogeneous states of the population~\cite{Bauer2017}. Depending on how growth rate differences between producers and non-producers (quantified by the selection strength $s$) balance with the response rate to the environment (quantified by the response probability $\lambda$), homogeneous (unimodal) or heterogeneous (bimodal) stationary densities are approached at long times in the mean-field equation~\eqref{eq:mean-field}.
The existence and the stability of heterogeneous stationary densities is a consequence of the feedback between ecological and population dynamics.
In total, the analysis of the mean-field equation~\eqref{eq:mean-field} shows that phenotypic heterogeneity arises dynamically in the quorum-sensing model and that it is robust both against changes in the definition of the stochastic many-particle process (how up-regulation and growth rate differences are implemented), and against perturbations and demographic noise of the stochastic dynamics.

\section{Main result of this work: Convergence to mean-field for large $N$}
\label{sec:main_result}

\paragraph{Purpose of this manuscript.}
In this manuscript, we prove that for any time $t>0$ the empirical density $\rho_N^{(1)}(t)$ of the stochastic many-particle process (microscopic dynamics / process) converges in probability towards the mean-field density  $\rho$ (macroscopic dynamics / process) as the number of individuals becomes large and if initial correlations are not too strong; see Figure~\ref{fig:outline_proof} for an overview. 
In other words, the mean-field equation~(\ref{eq:mean-field}) exactly describes the collective dynamics of the stochastic many-particle process of the quorum-sensing model as $N\to \infty$.

\paragraph{Closeness between microscopic and macroscopic process, and convergence in probability.} First, let us define the notion of closeness between the microscopic and the macroscopic process, and formulate our main result.
Since the empirical density $\rho^{(1)}_N$ is a a sum of delta functions (a histogram) while the mean-field density $\rho$ is a continuous function, closeness of the two can only hold in a weak sense. 
Consequently, we introduce a weak notion of distance between functionals (for example, $\rho^{(1)}_N$ and $\rho$) in $(L^\infty)^*$, which denotes the dual space of $L^\infty$ (the space of all essentially bounded measurable functions).
We then measure distances between probability distributions with the bounded Lipschitz metric, which is based on the bounded Lipschitz norm defined as follows. 
Defining the Lipschitz norm of a function $f\in C[0,1]$ (the space of all continuous functions on $[0,1]$) as:
\begin{align}\nonumber
	\|f\|_L\coloneqq\sup_{x,y\in[0,1]}\frac{|f(x)-f(y)|}{|x-y|}\ ,
\end{align}
the bounded Lipschitz norm of any functional $g\in(L^\infty)^*$ is given by:
\begin{align}\label{eq:BL_norm}
\|g\|_{BL}\coloneqq\sup_{\|f\|_L=1;f(0)=0}\left|\int_0^1 \d p\ f(p)g(p)\right|\ .
\end{align}
Note that $(L^\infty)^*$ can be identified with the space of all finitely additive finite signed measures. 
Furthermore, any normalized and positive $g\in(L^\infty)^*$ (such that $\int_0^1\d p\ g(p)=1$) can naturally be identified with a probability distribution. For such a normalized and positive $g$ one may drop the boundary condition $f(0)=0$ in the definition of the bounded Lipschitz norm.
Note that
\begin{align}\label{eq:BLL1}
\|g\|_{BL}\leq \left|\int_0^1\d p\ p|g(p)|\right|\leq\|g\|_1 =  \int_0^1\d p\ |g(p)|\ .
\end{align}
The bounded Lipschitz metric measures the distance between two functionals $g$ and $h\in(L^\infty)^*$ as:
\begin{align}\label{eq:BL_metric}
d(g,h) \coloneqq \|g-h\|_{BL} = \sup_{\|f\|_L=1;f(0)=0} \left| \int\d p\ f(p)g(p)-\int\d p\ f(p)h(p) \right|\ .
\end{align}

Furthermore, the convergence of the empirical density of the microscopic process $\rho^{(1)}_N$ against the solution $\rho$ of the macroscopic mean-field equation~\eqref{eq:mean-field} can at best hold in a probabilistic sense: With very small probability, always the same individual might be reproduced in a single realization of the stochastic process. Such a realization would lead to a big deviation from the solution of the mean-field equation~\eqref{eq:mean-field}. However, the occurrence of such a trajectory is improbable.
To capture this intuition in mathematical terms, we define convergence in probability as follows:
%
\begin{defn}\label{def:inprob}
Let $(\nu_{N})_N$ be a sequence of probability densities and $\nu$ be a probability density. 
We write:
\begin{align}\nonumber
\nu_{N}\tos\nu\ ,\quad \text{if for any } \epsilon>0:\
\lim_{N\to\infty} \mathbb{P}\left(d (\nu_{N},\nu)>\epsilon\right)=0\ .
\end{align}
\end{defn}
%
With this notion of convergence in probability, the main result of this manuscript is formulated as follows:
%
\begin{thm}\label{thm:main}
Let $\rho_N^{(1)}(t)$ be the empirical one-particle density of the stochastic many-particle process (the microscopic process) and $\rho(t)$ a solution of the mean-field equation~\eqref{eq:mean-field} (the macroscopic process). 
We assume for the initial densities that $\rho_N^{(1)}(0)\tos \rho(0)$.
Then: 
\begin{align}\nonumber
\rho_N^{(1)}(t) \tos  \rho(t) \quad \text{ for any } t>0\ .
\end{align} 
\end{thm}
%
It is not surprising, that one of the crucial steps in proving our result makes use of the law of large numbers. 
However, controlling the propagation of errors, which build up by neglecting correlations of the individuals' production degrees, with mathematical rigor is not trivial. 
The skeleton of our proof is summarized in Figure~\ref{fig:outline_proof} and outlined in the following.

%
\begin{figure}[htb!]
  \centering
  \includegraphics[width=\textwidth]{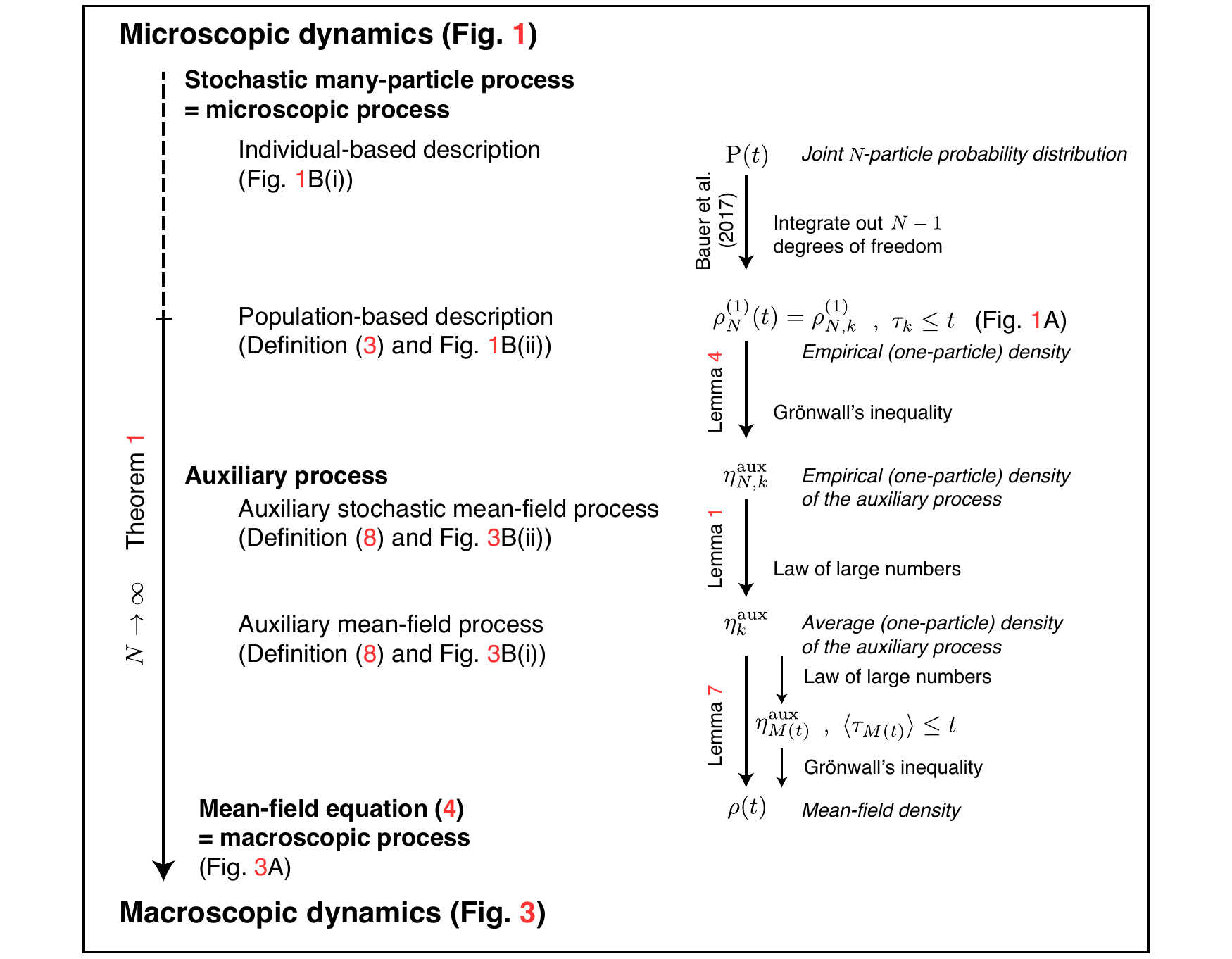}
  \caption{
Sketch of the main steps of the proof for the convergence towards mean-field in the quorum-sensing model. 
We prove that the microscopic description of the stochastic many-particle process (see Figure~\ref{fig:model}) converges to the macroscopic description of the quorum-sensing model (given by the mean-field equation~\eqref{eq:mean-field}) as $N\to\infty$.
More precisely, we establish that the empirical density of the microscopic process $\rho^{(1)}_N$ converges in probability to the macroscopic mean-field density $\rho$ as $N\to \infty$ if initial correlations are not too strong; see Theorem~\ref{thm:main}. 
The steps of the proof are summarized on the right hand side of the sketch.
The central idea is the introduction of an auxiliary  process, which mimics the time evolution of the mean-field equation as a stochastic process and updates the production degrees of the individuals in an independent manner between different update steps (``auxiliary stochastic mean-field process'' with probability density $\eta_{N,k}^\text{aux}$ at the $k^\text{th}$ update step). 
This way, arguments involving the law of large numbers can be separated from controlling the propagation of errors that build up due to correlations of the individuals' production degrees.
Along all arrows, we show weak convergence in probability (see Definition~\ref{def:inprob}). The central argument and the lemma, in which the respective convergence is proven, are written next to the according arrows. 
Empirical densities ($\rho_N^{(1)}$ and $\eta_N^\text{aux}$) are denoted by the subscript $N$ while average densities ($\rho$ and $\eta^\text{aux}$) do not carry a subscript.
We envision that our developed method to show convergence towards mean-field by introducing an auxiliary stochastic mean-field process may also  be helpful for other stochastic many-particle processes.
}
\label{fig:outline_proof}
\end{figure}
%

\paragraph{Outline of the proof.} 
The key idea of the proof is to separate the law of large numbers argument from the estimate of the error propagation by introducing an auxiliary stochastic mean-field process; see Figure~\ref{fig:macroscopic} for an illustration.
By virtue of the auxiliary process, individuals are created and annihilated in an explicitly independent manner between consecutive update steps such that the temporal evolution of the auxiliary process mimics the mean-field dynamics~\eqref{eq:mean-field} (see also Equation~\eqref{eq:mean-field_discrete}). 
The auxiliary process is characterized by the average density $\eta^\text{aux}$, and a single realization of the auxiliary process is denoted as $\eta^\text{aux}_N$ (the empirical density of the auxiliary stochastic mean-field process). 
Note that throughout the manuscript, we denote empirical densities ($\rho_N^{(1)}$ and $\eta_N^\text{aux}$) by the subscript $N$, whereas average densities ($\rho$ and $\eta^\text{aux}$) do not depend on the system size.

The idea of an auxiliary stochastic mean-field process with independent birth and death between consecutive update steps may seem paradox because the annihilation of an individual in the microscopic stochastic process always depends upon the actual state of the population: an individual with production degree $p$ can only be annihilated if it is existing. 
For the auxiliary process, however, we relax this condition of an actual existence of individuals: for the realization of the auxiliary process, we count an individual with a positive mass $+1$ at a birth event and an individual with a negative mass $-1$ at a death event. In other words, instead of creating/annihilating an individual, an individual with a positive/negative mass is created. 
This way, the empirical density ($\eta_N^\text{aux}$) may not be positive for all production degrees, but it is still normalized. 
The auxiliary process is implemented in such a way that we do not lose independence between consecutive update steps, see Definition~\ref{def:aux}.
As a consequence, the convergence of the empirical density of the auxiliary process $\eta^\text{aux}_N$ to the average density $\eta^\text{aux}$ as $N\to\infty$ is controlled with a ``standard'' law of large numbers argument (see Lemma~\ref{lem:LLN}). 

The propagation of errors for the convergence of the microscopic process ($\rho_N^{(1)}$) to the empirical auxiliary process ($\eta_N^\text{aux}$) is then controlled by Gr\"onwall's inequality (see Lemma~\ref{lem:micro}).
Gr\"onwall's inequality was also applied for the convergence of the average auxiliary process ($\eta^\text{aux}$) to the mean-field density $\rho$ (see Lemma~\ref{lem:macro}; the law of large numbers was implicitly applied as well).

Note also that most of the following sections do \textit{not} include the ``real'' time $t$ as a variable: the microscopic process, the auxiliary stochastic mean-field process, and the auxiliary mean-field process are synchronized in time. 
That is, we use $k = 0, 1, 2, \dots$ as a variable to count the ordinal number of creation/annihilation steps of the various processes; see Figure~\ref{fig:model}(A). In other words, $k$ labels the update steps.
Since the time intervals between two creation/annihilation processes are distributed independently, the discrete label $k$ is approximately a function of the continuous time $t$. 
Only when we compare the macroscopic process ($\rho$) with the auxiliary mean-field process ($\eta^\text{aux}$) is the time variable recovered, and convergence with respect to synchronization of time is controlled by a law of large numbers argument (see Lemma~\ref{lem:macro}).

%
\begin{rem}
The basic idea of the proof is to estimate the expectation value of the distance $d(\rho_N^{(1)}(t),\rho(t))$ as follows:
\begin{align}\nonumber
\mathbb{E}\left(d(\rho_N^{(1)}(t),\rho(t))\right)
\leq Const(t)\cdot\left(\mathbb{E}\left(d(\eta^\text{aux}_{N,0},\rho_{0})\right)+ N^{-1/4}\right)\ ,
\end{align}
with some constant $0<Const<\infty$ for any chosen time $t>0$.

Applying Markov's inequality then establishes an error estimate of the convergence for any $t>0$:
\begin{align}\nonumber
\mathbb{P}\left(d(\rho_N^{(1)}(t),\rho(t))>\epsilon_N\right)
&\leq \frac{\mathbb{E}\left(d(\rho_N^{(1)}(t),\rho(t))\right)}{\epsilon_N}\ ,\\ \nonumber
&\leq \frac{Const(t)}{\epsilon_N}\cdot\left(\mathbb{E}\left(d(\eta^\text{aux}_{N,0},\rho_{0})\right) +  N^{-1/4}\right)\ .
\end{align} 
This estimate provides a quantitative control of the propagation of errors with respect to the population size $N$.
For example, if initial correlations vanish with $N$ as $\mathbb{E}\left(d(\eta^\text{aux}_{N,0},\rho_{0})\right)<Const\cdot N^{-1/8}$, the choice $\epsilon_N = N^{-1/8}$ yields the estimate:
\begin{align}\nonumber
\mathbb{P}\left(d(\rho_N^{(1)}(t),\rho(t))>N^{-1/8}\right)\leq
Const(t) \cdot  N^{-1/8}\ .
\end{align} 
This statement quantifies our intuitive reasoning from above: 
Realizations of the stochastic many-particle process whose one-particle density deviate significantly ($d(\rho_N^{(1)}(t),\rho(t)) >  N^{-1/8}$) from the solution of the mean-field equation~\eqref{eq:mean-field} can actually occur (as is also seen in numerical simulations of the quorum-sensing model).
However, the probability $\mathbb{P}$ of such an occurrence is bounded from above by $\mathbb{P}\leq Const(t) \cdot  N^{-1/8}$. 
Pictorially speaking, as the population size $N$ grows, such occurrences become less and less likely and the magnitude of such deviations becomes smaller and smaller.

\end{rem}
%

\section{Introduction of the auxiliary stochastic mean-field process -- the central idea of the proof}
\label{sec:auxiliary_process}

\subsection{The microscopic, stochastic many-particle process in a population-based description}
\label{sec:microscopic_revisited}

First, the set-up of the microscopic process is revisited before we define the auxiliary stochastic mean-field process in Definition~\ref{def:aux}. In the following, we provide an alternative formulation of the stochastic many-particle process in a population-based description that is suitable for our proof; see Figure~\ref{fig:model}(B) for a comparison if the individual-based and population-based description. 

As described above, an individual $i$ reproduces at rate $\phi_i$ (for illustration, we assume the form of the fitness in Equation~\eqref{eq:fitness}) that depends on both the individual's production degree $p_i$ and the average production level in the whole population $\langle p \rangle$. 
The time taken until the next reproduction event of individual $i$ occurs is exponentially distributed with mean $\phi_i$. In other words, the waiting time is sampled from the probability density $\phi_i e^{-\phi_i t}$. 
For our purposes it is useful to reformulate the stochastic process in the spirit of Gillespie's stochastic kinetic Monte Carlo method~\cite{Gillespie1976,Gillespie1977}.
Instead of randomly choosing a time of reproduction for every individual independently, one can choose the time steps at which \textit{some} individual of the population is reproduced randomly and, in a second step, define another random variable that selects \textit{which} of the individuals reproduces. 
This reformulation does not change the dynamics of the microscopic, stochastic many-particle process.
For this reformulation, let $\tau_0=0$ and $\tau_k$ with $k= 1, 2, \dots$ be the time at which for the $k^{\text{th}}$ time the configuration of the population is updated (the $k^{\text{th}}$ update step), that is, for the $k^{\text{th}}$ time an individual is created (and another individual is annihilated at the same time). 
The total rate of creating any individual is given by the sum of the fitnesses of all individuals: $\sum_{i=1}^N\phi_i = N\langle \phi\rangle = N(1+s(b-c)\langle p\rangle)$.
It follows that all time differences $\Delta \tau_{k+1}=\tau_{k+1}-\tau_k$ are exponentially distributed with mean $\mu_k$. 
Because the configuration is updated at the time $\tau_k$, also the fitness and, thus, the parameter $\mu_k$ depend on the update step $k$. 
Writing $\langle p\rangle_{k}$ for the average production degree at time $\tau_k$ (that is, $\langle p\rangle_{k}=\frac{1}{N}\sum_{i=1}^N p_i(\tau_k)$), it follows that $\mu_k=N(1+s(b-c)\langle p\rangle_{k}) = N\langle \phi\rangle_k$; see Figure~\ref{fig:model}(A).
%
\begin{defn}\label{def:microtime}
 Let $\Delta\tau_{k+1}$ be the random variable for the length of the time interval between the $k^{\text{th}}$ and $(k+1)^{\text{th}}$ update step of the stochastic many-particle process of the quorum-sensing model, that is $\langle\Delta\tau_{k+1}\rangle\coloneqq 1/\mu_k = N^{-1} \left(1+s(b-c)\langle p\rangle_k\right)^{-1}$.  
We define $\tau_k\coloneqq \sum_{l=1}^{k-1} \Delta\tau_l$ as the update times and  $\langle \tau_k\rangle \coloneqq\sum_{l=1}^{k-1} \langle\Delta\tau_l \rangle$ as their according average. 
 Furthermore, we define $M(t)$ to be the maximal natural number such that $\langle \tau_{M(t)}\rangle \leq t$, and $\kappa(t)$ as the random variable given by the maximal number such that $\tau_{\kappa(t)}\leq t$.
Note that, due to the definition of the fitness, the number of update steps up to time $t$ scales linearly with $N$ on average, that is $M(t)\sim\mathcal{O}(N)$ such that $\langle \tau_{M(t)}\rangle\sim \mathcal{O}(N^0)$.
\end{defn}
%

Having defined the update times $\tau_k$ for $k\geq 1$, we next define the random variables that select two new individuals with production degree $p^{\text{micro},+}_{N+2k-1}$ and $p^{\text{micro},+}_{N+2k}$ for creation, and two new individuals with production degree $p^{\text{micro},-}_{2k-1}$ and $p^{\text{micro},-}_{2k}$ for annihilation at time $\tau_k$. 
The random variables $p^{\text{micro},+}_{N+2k-1}$, $p^{\text{micro},+}_{N+2k}$,  $p^{\text{micro},-}_{2k-1}$, and $p^{\text{micro},-}_{2k}$ map from some probability space $\Omega^k$ onto the interval $[0,1]$ at update step $k$ for all $k \geq 1$. Before we define the probability space $\Omega^k$, we first define the probability distribution of production degrees that we consider in the microscopic process.


The population-based description of the microscopic process begins with a set of $N$ individuals and their according production degrees $p^{\text{micro},+}_1,\dots, p^{\text{micro},+}_N$.
The sequence of production degrees that were initially present and that have been created until time $\tau_k$ are denoted as $P^{\text{micro},+}_k\coloneqq(p^{\text{micro},+}_1,\dots,p^{\text{micro},+}_N, p^{\text{micro},+}_{N+1}, \dots,p^{\text{micro},+}_{N+2k})$; the  sequence of production degrees that have been annihilated until $\tau_k$ is denoted as $P^{\text{micro},-}_k\coloneqq(p^{\text{micro},-}_1,\dots,p^{\text{micro},-}_{2k})$. 
%
\begin{defn}\label{defn:empdens}
For any pair of sequences $P^{\text{micro},+}_k$, $P^{\text{micro},-}_k$, the empirical one-particle density of the microscopic process after $k$ update steps is given by:
\begin{align}\nonumber
\rho^{(1)}_{N,k}(p) &= \rho^{(1)}_N(p; P^{\text{micro},+}_k, P^{\text{micro},-}_k)\\
&\coloneqq\frac{1}{N}\left(\sum_{j=1}^{N+2k}\delta(p-p^{\text{micro},+}_j)-\sum_{j=1}^{2k}\delta(p-p^{\text{micro},-}_j)\right)\ .
\end{align}
\end{defn}
%
Note that $\rho^{(1)}_{N,k}$ is positive and fulfils $\int \d p\ \rho^{(1)}_{N,k}=1$ for all update steps $k$. Thus, $\rho^{(1)}_{N,k}$ has the form of a probability distribution for all $k$. 
Note also that only existing particles may be annihilated in the microscopic process. 
On the other hand, in the auxiliary process (see Definition~\ref{def:aux} below), individuals may be created with a negative mass at any production degree according to the present density of particles. Therefore, the creation and annihilation of individuals at a certain production degree is independent of the previous existence of individuals at that production degree in the auxiliary process. 
This way, positivity of the empirical density ($\eta_N^\text{aux}$) is lost for the auxiliary process, but the normalization is still valid. 

For easier comparison of the random variables $p^{\text{micro},+}$ and $p^{\text{micro},-}$ of the  microscopic model with the yet to be defined random variables of the auxiliary model (see Definition~\ref{def:aux}), it is convenient to assume a constant probability density on $\Omega^k$ and choose the maps $p^{\text{micro},+}$ and $p^{\text{micro},-}$ in such a way that the creation and annihilation probabilities coincide with those of the microscopic process. 

To sample a random variable from an arbitrary probability density $\nu$, we use the following definition: 
%
\begin{defn}\label{def:rv}
For any probability density $\nu\in(L^\infty)^\star$, we define the random variable $X^\nu$, with  $[0,1]\to X^\nu$, through the so-called quantile function:
\begin{align}
X^\nu(\theta)\coloneqq\inf\left\{x:\int_{0}^x \d p\ \nu(p) >\theta \right\} \ .
\end{align}
Note that, because probability distributions functions are continuous from the right, the infimum is in fact a minimum. 
Thus, the random variable $X^\nu(\theta)$, with $\theta$ being uniformly distributed on $[0,1]$, is the inverse function to the cumulative distribution function of $\nu$ that is given by $x\mapsto \int_{0}^x\d p\ \nu(p)$; see also Figure~\ref{fig:Laisant} (upper part) for an illustration. 
\end{defn}
%
Following this definition, the random variable $X^\nu(\theta)$, with $\theta$ being uniformly distributed on $[0,1]$, has probability density $\nu$.
For a given probability density $\nu$, we also define the reproduction density that accounts for how production degrees change at an update step of the stochastic many-particle process defined by the quorum-sensing model.
%
\begin{defn}\label{def:Phi}
For any probability density $\nu:[0,1]\to\mathbb{R}^+_0$, let $\Phi(\nu)$ be the reproduction density of the quorum-sensing model, that is the probability density given by:
\begin{align}\nonumber
\Phi(\nu)(p)\coloneqq 
2\lambda \delta\left(p-R(\langle p \rangle_\nu)\right)
+    (1-2\lambda)\frac{1+s\big(b\langle p \rangle_\nu-cp\big)}{1+s(b-c)\langle p \rangle_\nu}\nu(p)\ ,  
\end{align}
where we abbreviated the mean of $\nu$ as $\langle p \rangle_\nu=\int_{0}^1\d p\ p\ \nu(p)$ and used the fitness function in Equation~\eqref{eq:fitness}. 
\end{defn}
%
The reproduction density consists of two parts: (i) response to the average by increasing probability mass at the production degree $R(\langle p \rangle_\nu)$ with prefactor $2\lambda$ (sense-and-response term), and (ii) reproduction according to relative fitness differences $\phi(p)/\phi(\overline{p})\cdot\nu(p)$ with prefactor $(1-2\lambda)$ (replicator term). 
For later purposes, we mention that the notion of the reproduction density $\Phi$ facilitates to rewrite the mean-field equation in discrete time steps as follows (compare with linearization of the mean-field equation~\eqref{eq:mean-field}):
\begin{equation}\label{eq:mean-field_discrete}
\rho_{\tau_{k+1}}(p) =  \rho_{\tau_{k}}(p) + \frac{1}{N}\frac{\Delta\tau_{k+1}}{\langle \Delta\tau_{k+1}\rangle_{\rho_{\tau_{k}}}}\big(\Phi(\rho_{\tau_{k}})(p)-\rho_{\tau_{k}}(p)\big)+\mathcal{O}(\Delta\tau_{k+1}^2).
\end{equation}

With these definitions, the microscopic, stochastic many-particle process is reformulated as follows.

%
\begin{defn}Let the sample space $\Omega$ be given by the sequence $\Omega\coloneqq(\Omega^1,\Omega^2,\ldots)$ where the individual sample spaces $\Omega^k$ at update step $k$ are given by  $\Omega^k=(\Omega^k_1,\Omega^k_2,\Omega^k_3,\Omega^k_4)$ with $\Omega^k_1=\Omega^k_2=[0,1]$ and $\Omega^k_3=\Omega^k_4=\{0,1\}$.
We assume that all of the $\Omega^k_j$ are independent (both in the indices $j$ and $k$), that $\omega^k_1$ and $\omega^k_2$ are uniformly distributed on $[0,1]$, and that $\mathbb{P}(\omega^k_3=0)=\mathbb{P}(\omega^k_4=0)=\lambda$. We write $\omega^k = (\omega^k_1, \omega^k_2, \omega^k_3, \omega^k_4)$.
\end{defn} 
%

We now use this sample space to reformulate the microscopic process of the quorum-sensing model. 
In the $k^\text{th}$ update step $\omega^k_1$ and $\omega^k_2$ determine the two individuals that are subsequently annihilated (with production degrees $p_{2k-1}^{\text{micro},-}$ and $p_{2k}^{\text{micro},-}$), and $\omega^k_3$ and $\omega^k_4$ determine the production degrees of the two created individuals ($p_{N+2k-1}^{\text{micro},+}$ and $p_{N+2k}^{\text{micro},+}$).
If $\omega^k_{3/4}=0$, then the first/second newly created individual attains the production degree given by the value $R(\langle p\rangle_{\nu})$; if $\omega^k_{3/4}=1$ then the first/second newly created individual takes over the production degree of the first annihilated individual, that is, it attains the production degree $p_{2k-1}^{\text{micro},-}$; see  Figure~\ref{fig:model}(B)(ii) for an illustration.
Using this sample space and Definition~\ref{def:rv}, the microscopic process of the quorum-sensing model can be reformulated as follows to reproduce the correct probability distribution of the individuals' production degrees:
%
\begin{defn}\label{defn:micro_process}
Let $p_{j}^{\text{micro},+}$ be the initial production degree of the $j^{\text{th}}$ individual for $1\leq j\leq N$.
The random variables $p_{N+2k}^{\text{micro},+}$ and $p_{N+2k-1}^{\text{micro},+}$ denote the values of the production degrees of the two individuals that are created in the $k^{\text{th}}$ update step, and the random variables $p_{2k-1}^{\text{micro},-}$ and $p_{2k}^{\text{micro},-}$ denote the values of the production degrees of the two individuals that are annihilated in the $k^{\text{th}}$ update step. 
These random variables are given by: 
\begin{align}\nonumber
p_{2k-1}^{\text{micro},-}(\omega^k)&\coloneqq X^{\Phi(\rho^{(1)}_{N,k-1})}(\omega^k_1)\ ,\\\nonumber
p_{2k}^{\text{micro},-}(\omega^k)&\coloneqq X^{\rho^{(1)}_{N,k-1}-\delta(p-p_{2k-1}^{\text{micro},-}(\omega^k_1))}(\omega^k_2)\ ,\\\nonumber
p_{N+2k-1}^{\text{micro},+}(\omega^k)&\coloneqq p_{2k-1}^{\text{micro},-}\omega^k_3 + R(\langle p \rangle_{\rho^{(1)}_{N,k-1}}) (1-\omega^k_3)\ ,  \\\nonumber
p_{N+2k}^{\text{micro},+}(\omega^k)&\coloneqq p_{2k-1}^{\text{micro},-}\omega^k_4 + R(\langle p\rangle_{\rho^{(1)}_{N,k-1}}) (1-\omega^k_4)  \;.
\end{align}
Let $P^{\text{micro},+}_k\coloneqq(p^{\text{micro},+}_1,\dots,p^{\text{micro},+}_N, p^{\text{micro},+}_{N+1}, \dots,p^{\text{micro},+}_{N+2k})$ and  $P_k^{\text{micro},-}\coloneqq(p_1^{\text{micro},-},p^{\text{micro},-}_2,\dots,p_{2k}^{\text{micro},-})$.
Together with Definition~\ref{defn:empdens}, the empirical density of the microscopic process after $k = 1, 2, \dots$ update steps follows as:
\begin{align}\label{eq:empirical_one-particle_density}
\rho^{(1)}_{N,k}(p)= \rho^{(1)}_N(p; P^{\text{micro},+}_k, P^{\text{micro},-}_k)\;.
\end{align}
\end{defn}
%
The values of the random variables at the update step $k$ depend upon the probability distribution of production degrees at the update step $k-1$. 
The definition above assures that only individuals present at $\tau_{k-1}$ can be chosen for annihilation and, thus, to inherit their production degree $p_{2k-1}^{\text{micro},-}$. 
It follows by induction that $\rho^{(1)}_{N,k}$ is in fact positive for all update steps $k$, as claimed above. 
 
Note that $\rho^{(1)}_{N}(t)$ denotes the empirical  density at time $t$, and $\rho^{(1)}_{N,k}$ denotes the empirical density after $k$ update steps of the coupled creation-annihilation (birth-death) process.
Thus, with Definition~\ref{def:microtime}, it follows that $\rho^{(1)}_N(t)=\rho^{(1)}_{N,\kappa(t)}$. 
%
In contrast to the random variables $\rho^{(1)}_N(t)$ and $ \rho^{(1)}_{N,\kappa(t)}$,   $\rho (t)$ is a probability density. Given that both the number of individuals is large and the dependence between the individuals is mild, one expects that $ \rho^{(1)}_{N,\kappa(t)}$ converges in probability to $\rho (t)$ as stated in Theorem~\ref{thm:main}. 

\subsection{Definition of the auxiliary stochastic mean-field process}

We now define the auxiliary stochastic mean-field process. 
Heuristically speaking, the temporal evolution of the production degrees of the population in the auxiliary process mimic the mean-field dynamics defined by Equation~\eqref{eq:mean-field}; see also its discretized form~\eqref{eq:mean-field_discrete}. 
The central idea of the proof for the convergence of mean-field is to construct the auxiliary process in such a way that (i) the production degrees of the individuals at one update step are created with positive and negative masses, and that (ii) these masses are sampled \textit{independently} of the realization of the previous update step (in contrast to the microscopic process).
At one update step two particles of positive mass and two particles of negative mass are created, which are correlated in general.
Important for our purpose is the independence between consecutive update steps.
The respective random variables are given by $p_{N+2k-1}^{\text{aux},+}$ and $p_{N+2k}^{\text{aux},+}$ for the created production degrees, and $p_{2k-1}^{\text{aux},-}$ and $p_{2k}^{\text{aux},-}$ denote the annihilated production degrees at the $k^\text{th}$ update step of the auxiliary process.
In total, the auxiliary process is defined as follows:
%
\begin{defn}\label{def:aux}
The average density of the auxiliary process at the ($k+1)^{\text{th}}$ update step follows from the average density at update step $k$ as:
\begin{align}\nonumber
	\eta^\text{aux}_{k+1}=\eta^\text{aux}_{k}-\frac{1}{N}\eta^\text{aux}_{k}+ \frac{1}{N}\Phi(\eta^\text{aux}_{k})\ ,
\end{align}
and mimics the temporal evolution of the mean-field equation~\eqref{eq:mean-field} at discretized time steps; see Equation~\eqref{eq:mean-field_discrete}. 
The initial probability distribution of production degrees in the population for the auxiliary process is given by $\eta^\text{aux}_{0}=\rho_0$.
Importantly, the update of $\eta^\text{aux}$ is independent of the realization of the auxiliary process at the previous update step, whereas in the microscopic process the time evolution depends upon the realization of the stochastic process.

In one realization of the auxiliary stochastic mean-field process at the update step $k$, individuals with negative and positive masses are created independently of the production degrees of the individuals present at update step $k-1$. The respective random variables describing the values of the production degrees created and annihilated at the $k^{\text{th}}$ update step are given by (two individuals are created and two individuals are annihilated per update step): 
\begin{align}\nonumber
p_{2k-1}^{\text{aux},-}(\omega^k)&\coloneqq X^{\Phi(\eta^\text{aux}_{k})}(\omega^k_1)\ ,\\\nonumber
p_{2k}^{\text{aux},-}(\omega^k)&\coloneqq X^{\eta^\text{aux}_{k}}(\omega^k_2)\ ,\\\nonumber
p_{N+2k-1}^{\text{aux},+}(\omega^k)&\coloneqq p_{2k-1}^{\text{aux},-}\omega^k_3 + R(\langle p\rangle_{\eta^\text{aux}_{k}}) (1-\omega^k_3)\ ,\\\nonumber
p_{N+2k}^{\text{aux},+}(\omega^k)&\coloneqq p_{2k-1}^{\text{aux},-}\omega^k_4 + R(\langle p \rangle_{\eta^\text{aux}_{k}}) (1-\omega^k_4)\ .
\end{align}

The values of $p_{N+2k-1}^{\text{aux},+}$ and $p_{N+2k}^{\text{aux},+}$ depend both on $p_{2k-1}^{\text{aux},-}$ and, thus, indirectly also on each other. Therefore, they are not independent. This dependence, however, is not problematic for our proof because independence holds still true for the vast majority of the individuals' production degrees.
The empirical density of the auxiliary process after $k$ update steps is given by:
\begin{align}\nonumber
\eta^{\text{aux}}_{N,k}(p) \coloneqq\frac{1}{N}\left(\sum_{j=1}^{N+2k}\delta(p-p^{\text{aux},+}_j)-\sum_{j=1}^{2k}\delta(p-p^{\text{aux},-}_j)\right)\ .
\end{align}
Note that, through this definition, $\eta_N^\text{aux}$ may not be positive, but is always normalized.

\end{defn}
%

%
\begin{figure}[htb!]
  \centering
  \includegraphics[width=\textwidth]{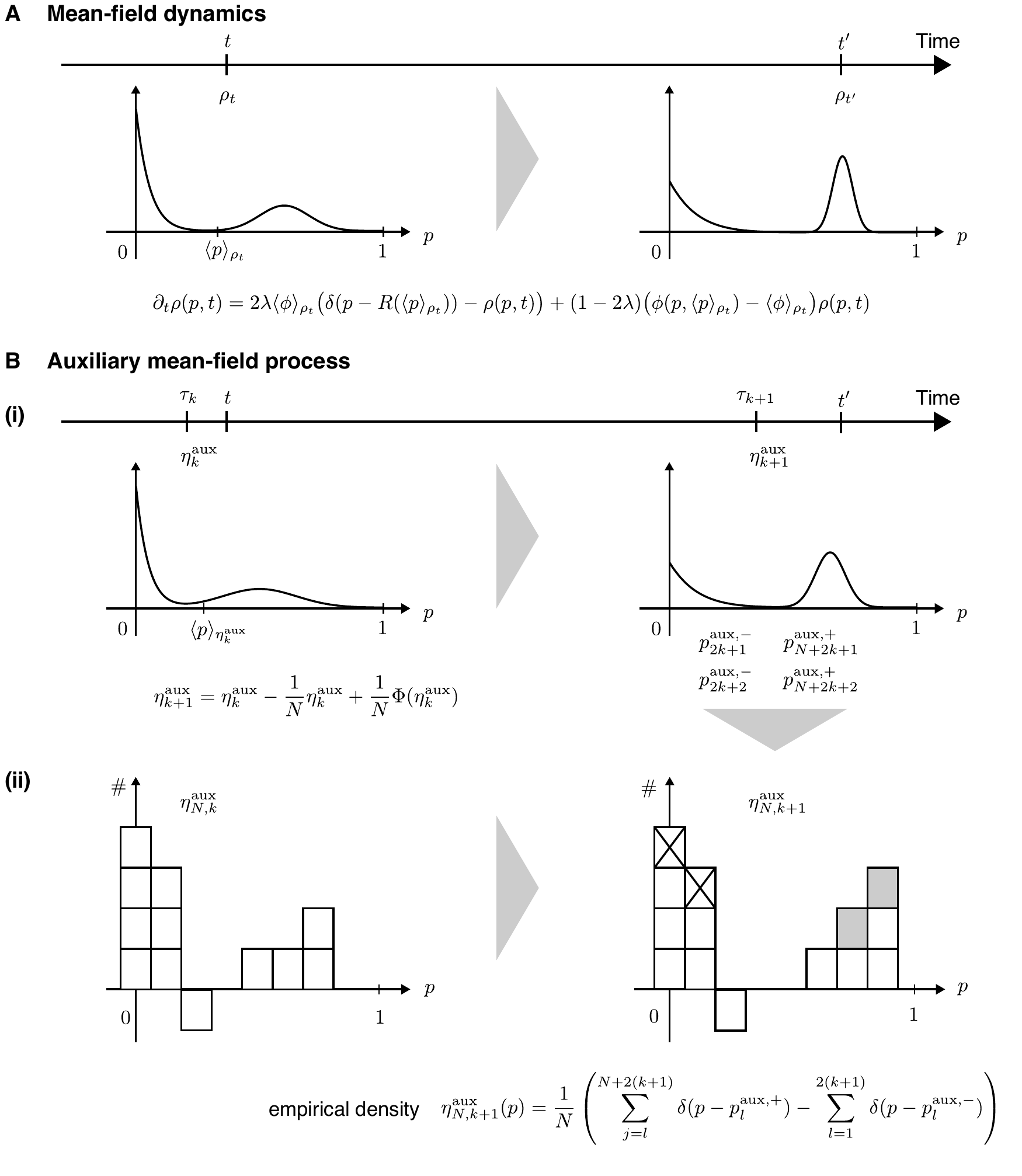}
  \caption{
  (A) Mean-field dynamics. The the mean-field density evolves continuously in time   according to mean-field equation~\eqref{eq:mean-field}.
  (B) Auxiliary mean-field process to prove convergence between microscopic and macroscopic dynamics. 
  (i) The auxiliary mean-field process mimics the temporal evolution of the mean-field equation at discrete time steps; compare with the discretized mean-field equation~\eqref{eq:mean-field_discrete}. In this exemplary realization, the last update steps  $\tau_k$ and  $\tau_{k+1}$ are depicted before time $t$ and $t'$, respectively. 
  For the auxiliary stochastic mean-field process, four masses ($p_{2k+1}^{\text{aux},-}$, $p_{2k+2}^{\text{aux},-}$, $p_{N+2k+1}^{\text{aux},+}$, and $p_{N+2k+2}^{\text{aux},+}$) are sampled at the update step $\tau_k$ by using the average density $\eta_k^\text{aux}$ following Definition~\ref{def:aux}. 
 (ii) With these sampled masses, the empirical density of the auxiliary process is updated from $\eta^{\text{aux}}_{N, k}$ to $\eta^{\text{aux}}_{N, k+1}$. 
  Note that the empirical density of the auxiliary process can be non-positive. 
  Furthermore, the masses are sampled independently of the realization of the previous update step (in contrast to the microscopic process).
}
  \label{fig:macroscopic}
\end{figure}
%

\section{Proof of the theorem for the convergence to mean-field}
\label{sec:proof}
\subsection{Convergence of the auxiliary process ($\eta^\text{aux}_{N,K} \tos \eta^\text{aux}_{K}$)  -- Law of large numbers argument}

Because individuals are created and annihilated in an independent manner between consecutive update steps in the auxiliary process, one expects that the empirical density $\eta^\text{aux}_{N,k}$ converges to the average density $\eta^\text{aux}_{k}$ of the auxiliary process for every update step $k$ as $N\to \infty$.
Here we show that indeed $\eta^\text{aux}_{N,k} \tos \eta^\text{aux}_{k}$.
More precisely, we have Lemma~\ref{lem:LLN}.

\begin{lem} \label{lem:LLN}
One finds a constant~$0<Const<\infty$ such that for a given update step $K$ the expected difference between a single realization of the auxiliary stochastic mean-field process ($\eta^\text{aux}_{N,K}$) and the average density of the auxiliary mean-field process ($\eta^\text{aux}_{K}$) is estimated as:
\begin{align}\nonumber
\mathbb{E}\left(d(\eta^\text{aux}_{N,K},\eta^\text{aux}_{K})\right)
\leq Const\cdot \frac{K^{3/4}}{N}
+\mathbb{E}\left(d(\eta^\text{aux}_{N,0},\eta^\text{aux}_{0})\right)\ .
\end{align}
\end{lem}

\begin{proof}
Note, that we are dealing with the auxiliary process and, thus, have independence of the random variables for the created and annihilated production degrees between consecutive update steps.
The proof of Lemma~\ref{lem:LLN} is based on a law of large numbers argument. Such an argument is standard, of course. However since the proof of the lemma is short and we deal with a special notion of the bounded Lipschitz distance $d(\cdot,\cdot)$, we provide it in the following. 

We first split the interval $[0,1]$ into $n$ pieces $I_j\coloneqq[\frac{j-1}{n},\frac{j}{n}]$ with $1\leq j\leq n$. 
Later on, $n$ is chosen as a function of the total number of update steps $K$.
The definition of $d(\cdot,\cdot)$ involves taking a supremum. Taking the supremum does not commute with taking the expectation value.
Therefore, we first estimate the distance $d(\eta^\text{aux}_{N,K},\eta^\text{aux}_{K})$ of the empirical density $\eta^\text{aux}_{N,K}$ from the average density $\eta^\text{aux}_{K}$ and take the expectation value later.

Using Lipschitz continuity of $f$ on every interval $I_j$ (that is $f(p) \leq |f(\frac{j-1}{n})|+ |\frac{j-1}{n}-p|$ because $\|f\|_L=1$ on every interval $I_j$), one obtains
\begin{align}\nonumber
d(\eta^\text{aux}_{N,K}&,\eta^\text{aux}_{K})
=\|\eta^\text{aux}_{N,K}-\eta^\text{aux}_{K}\|_{BL} \\\nonumber
&= \sup_{\|f\|_L=1} \left| \int_0^1\d p\ f(p)\eta^\text{aux}_{N,K}(p)-\int_0^1\d p\ f(p)\eta^\text{aux}_{K}(p) \right|\ ,\\\nonumber
&\leq\sum_{j=1}^n\sup_{\|f\|_L=1}\left|\int_{I_j} \d p\ f(p)\left(
\eta^\text{aux}_{N,K}(p)-\eta^\text{aux}_{N,0}(p)-\eta^\text{aux}_{K}(p)+\eta^\text{aux}_{0}(p)\right)  \right|+d(\eta^\text{aux}_{N,0},\eta^\text{aux}_{0})\ ,\\\nonumber
&\leq \sum_{j=1}^n\int_{I_j}\d p\ \left|f\left(\frac{j-1}{n}\right)\right|
\left|\eta^\text{aux}_{N,K}(p)-\eta^\text{aux}_{N,0}(p)-\eta^\text{aux}_{K}(p)+\eta^\text{aux}_{0}(p)\right|\\\nonumber
&\quad+\sum_{j=1}^n \int_{I_j}\d p\ \left|\frac{j-1}{n}-p\right|
\left|\eta^\text{aux}_{N,K}(p)-\eta^\text{aux}_{N,0}(p)-\eta^\text{aux}_{K}(p)+\eta^\text{aux}_{0}(p)\right|  +d(\eta^\text{aux}_{N,0},\eta^\text{aux}_{0})\ ,\\\nonumber
&\leq Const\cdot \left(1+\frac{1}{n}\right)\sum_{j=1}^n\int_{I_j}\d p\
\left|\eta^\text{aux}_{N,K}(p)-\eta^\text{aux}_{N,0}(p)-\eta^\text{aux}_{K}(p)+\eta^\text{aux}_{0}(p)\right|+d(\eta^\text{aux}_{N,0},\eta^\text{aux}_{0})\ .
\end{align}
It follows that
\begin{align}\nonumber
\mathbb{E}&(d(\eta^\text{aux}_{N,K},\eta^\text{aux}_{K}))\\\label{eq:estimate_d_aux_N}
&\leq Const\cdot\left(1+\frac{1}{n}\right)\sum_{j=1}^n \mathbb{E}\left(\int_{I_j}\d p\ \left|
(\eta^\text{aux}_{N,K}(p)-\eta^\text{aux}_{N,0}(p)-\eta^\text{aux}_{K}(p)+\eta^\text{aux}_{0}(p))\right|\right)+\mathbb{E}\left(d(\eta^\text{aux}_{N,0},\eta^\text{aux}_{0})\right)\ .
\end{align}

For each interval $I_j$ we now give a law of large numbers argument. 
We define the random variable $Y^{j,+}_{k,1}$ that takes value 1 if the individual $N+2k-1$ with a positive mass is sampled inside the interval $I_j$ in the $k^{\text{th}}$ update step, and  that takes value 0 otherwise.
Accordingly, the random variable $Y^{j,+}_{k,2}$  takes value 1 if individual $N+2k$  is sampled inside the interval $I_j$.
Furthermore, the random variables $Y^{j,-}_{k,1/2}$ indicate whether an individual is created with a negative mass in the interval $I_j$ at the $k^{\text{th}}$ update step:
\begin{align}\nonumber
Y^{j,-}_{k,1}(\omega^k)&\coloneqq
\begin{cases} 1 &\mbox{if } p_{2k-1}^{\text{aux},-}(\omega^k)\in I_j\ , \quad \text{that is, if }  X^{\Phi(\eta^\text{aux}_{k})}(\omega^k_1) \in I_j\ , \\
0 & \mbox{else}\ .  
\end{cases}\\\nonumber
Y^{j,-}_{k,2}(\omega^k)&\coloneqq
\begin{cases} 1 &\mbox{if } p_{2k}^{\text{aux},-}(\omega^k)\in I_j\ , \quad \text{that is, if }  X^{\eta^\text{aux}_{k}}(\omega^k_2)\in I_j  \ , \\
0 & \mbox{else}\ .  
\end{cases}\\\nonumber
Y^{j,+}_{k,1}(\omega^k)&\coloneqq
\begin{cases} 1 &\mbox{if } p_{N+2k-1}^{\text{aux},+}(\omega^k) \in I_j\ , \\
0 & \mbox{else}\ .  
\end{cases}\\\nonumber
Y^{j,+}_{k,2}(\omega^k)&\coloneqq
\begin{cases} 1 &\mbox{if } p_{N+2k}^{\text{aux},+}(\omega^k) \in I_j\ , \\
0 & \mbox{else}\ .  
\end{cases}
\end{align}
By Definition~\ref{def:aux} of the auxiliary process, the $Y^{j,\pm}_k$ are independent for different values of $k$, that is consecutive updates with birth and death are independent.
Therefore, the difference between positive and negative masses in the interval $I_j$ after $K$ update steps in one realization of the auxiliary process is obtained as:
\begin{align}\nonumber
	&\sum_{k=1}^K \left(Y^{j,+}_{k,1}(\omega^k)+Y^{j,+}_{k,2}(\omega^k )-Y^{j,-}_{k,1}(\omega^k)-Y^{j,-}_{k,2}(\omega^k)\right) = N\int_{I_j}\d p\ \left(\eta^\text{aux}_{N,K}(p) - \eta^\text{aux}_{N,0}(p)\right)\ .
\end{align}
%
By the definition of the average density $\eta^\text{aux}$ of the auxiliary process, it is:
\begin{align}\nonumber
	&\mathbb{E}\left(\sum_{k=1}^K \left(Y^{j,+}_{k,1}(\omega^k)+Y^{j,+}_{k,1}(\omega^k)-Y^{j,-}_{k,1}(\omega^k)-Y^{j,-}_{k,1}(\omega^k)\right)\right)= N\int_{I_j}\d p\ \left(\eta^\text{aux}_{K}(p) - \eta^\text{aux}_{0}(p)\right)\ .
\end{align}
Introducing  
$\left(Z_k\right)_{k\in\{1,\ldots,K\}}\in\left\{\left(Y^{j,+}_{k,1}\right)_{k\in\{1,\ldots,K\}},\left(Y^{j,+}_{k,2}\right)_{k\in\{1,\ldots,K\}},\left(Y^{j,-}_{k,1}\right)_{k\in\{1,\ldots,K\}},\left(Y^{j,-}_{k,2}\right)_{k\in\{1,\ldots,K\}}\right\}$, and using independence between the different update steps, we have a law of large numbers argument for every interval $I_j$ 
as follows:
\begin{align}\nonumber
\mathbb{E}\left(\left|\frac{1}{N}\sum_{k=1}^K Z_k
-\mathbb{E}\left(\frac{1}{N}\sum_{k=1}^KZ_k \right)\right|\right) 
\leq \left(\text{Var}\left( \frac{1}{N}\sum_{k=1}^KZ_k\right)\right)^{1/2}\leq \frac{1}{N} \sqrt{K} \frac{1}{2} \ .
\end{align}
The last estimate exploits the independence of random variables between consecutive steps of the sampling process, and the boundedness of the variance with $\text{Var}(Y^{j,\pm}_{k, 1/2})\leq 1/4$ for all $k = 1, \dots, K$.

Using triangle inequality and linearity of the expectation value we obtain:
\begin{align}\nonumber
\mathbb{E}&\left(\left|\frac{1}{N}\sum_{k=1}^K \left(Y^{j,+}_{k,1}+Y^{j,+}_{k,2}-Y^{j,-}_{k,1}-Y^{j,-}_{k,2}\right)
-\mathbb{E}\left(\frac{1}{N}\sum_{k=1}^K\left(Y^{j,+}_{k,1}+Y^{j,+}_{k,2}-Y^{j,-}_{k,1}-Y^{j,-}_{k,2}\right)\right)\right|\right) 
\leq \frac{2\sqrt{K}}{N}  \ .
\end{align}

Therefore, we obtain with Equation~\eqref{eq:estimate_d_aux_N}:
\begin{align}\nonumber
\mathbb{E}\left(d(\eta^\text{aux}_{N,K},\eta^\text{aux}_{K})\right)
&\leq Const\cdot \left(1+\frac{1}{n}\right)\sum_{j=1}^n \frac{2\sqrt{K}}{N}
+\mathbb{E}\left(d(\eta^\text{aux}_{N,0},\eta^\text{aux}_{0})\right)\ .
\end{align}
Choosing $n = K^{1/4}$ yields the estimate:
\begin{align}\nonumber
\mathbb{E}\left(d(\eta^\text{aux}_{N,K},\eta^\text{aux}_{K})\right)
\leq Const\cdot  \frac{K^{3/4}}{N}
+\mathbb{E}\left(d(\eta^\text{aux}_{N,0},\eta^\text{aux}_{0})\right)\ ,
\end{align}
which proves the lemma. In summary, the estimate for the expected distance between empirical and average density of the auxiliary process scales with the number of time steps $K$ as $K^{3/4}$: the factor $K^{1/4}$ stems from the chosen number of intervals and $\sqrt{K}$ stems from the law of large numbers on each of these intervals.
\end{proof}

\subsection{Convergence of the microscopic to the auxiliary process ($\rho^{(1)}_{N,K}\tos\eta^\text{aux}_{N,K}$) -- Control of error propagation with Gr\"onwall's inequality}

We now show that the propagation of errors, which build up over time due to the correlation of production degrees, can be controlled with Gr\"onwall's inequality.
In other words, the empirical density of the microscopic process $\rho^{(1)}_{N,K}$ converges to the empirical density of the auxiliary process  $\eta^\text{aux}_{N,K}$, as $N\to \infty$ for any finite update step $K$, see Lemma~\ref{lem:micro}. 

\begin{lem}\label{lem:aux1} 
Let $\nu\in(L^\infty)^\star$ and $\psi\in(L^\infty)^\star$ be two one-particle probability densities, $f$ some globally Lipschitz continuous function on $[0,1]$. Then
\begin{align}\nonumber
	\left|\langle f\rangle_{\nu}-\langle f\rangle_{\psi}\right|\leq \| f \|_L d(\nu,\psi)\ .
\end{align}
Here  $\langle\cdot\rangle_{\nu}$ and  $\langle\cdot\rangle_{\psi}$ means averaging with respect to $\nu$ and $\psi$, respectively, and $\|f\|_L$ is the global Lipschitz constant of~$f$.  
\end{lem}

\begin{proof}
Plugging in the definitions, one obtains:
\begin{align}\nonumber
\left|\langle f\rangle_{\nu}-\langle f\rangle_{\psi}\right|
&=\left|\int_0^1\d p\ f(p)\left(\nu(p)-\psi(p)\right) \right|\ =\|f\|_L\left|\int_0^1\d p\ \frac{f(p)}{\|f\|_L}\left(\nu(p)-\psi(p)\right) \right|\ .
\end{align}
Since $\left\| \frac{f(p)}{\|f\|_L}\right\|_L=1$ we can use it to test the supremum in the definition of the bounded Lipschitz distance $d(\cdot,\cdot)$ and obtain that the right hand side of the last equation is indeed bounded by $\|f\|_L d(\nu,\psi)$. 
\end{proof}

\begin{lem}\label{lem:Xbyd}
Let $\nu\in(L^\infty)^\star$ and $\psi\in(L^\infty)^\star$ be two one particle probability densities. 
Then 
\begin{align}\nonumber
\mathbb{E}(|X^{\nu}-X^{\psi}|)\leq d(\nu,\psi)\ ,
\end{align}
see Definition~\ref{def:rv} of the quantile function $X^{\nu}$ and $X^{\psi}$.
\end{lem}

\begin{proof}
Laisant's formula for inverse functions states that for any invertible function $g$, it is
\begin{align}\nonumber	
\int_{a}^{b}\d y\ g^{-1}(y)+\int_c^d\d x\ g(x)=bd-ac\ .
\end{align}
Applying Laisant's formula to the random variables $X^{\nu}$ and $X^{\psi}$, it follows that
\begin{align}\label{eq:Laisant}	
\mathbb{E}\left(\left|X^{\nu}-X^{\psi}\right|\right) = \int_0^1\d \theta\  \left|X^{\nu}(\theta)-X^{\psi}(\theta)\right| = \int_0^1\d p\ \left|\int_0^p\d y\ (\nu(y)- \psi(y))\right|\ ,
\end{align}
see Figure~\ref{fig:Laisant}(A) for a sketch.

\begin{figure*}[th!]
  \centering
  \includegraphics[width=\textwidth]{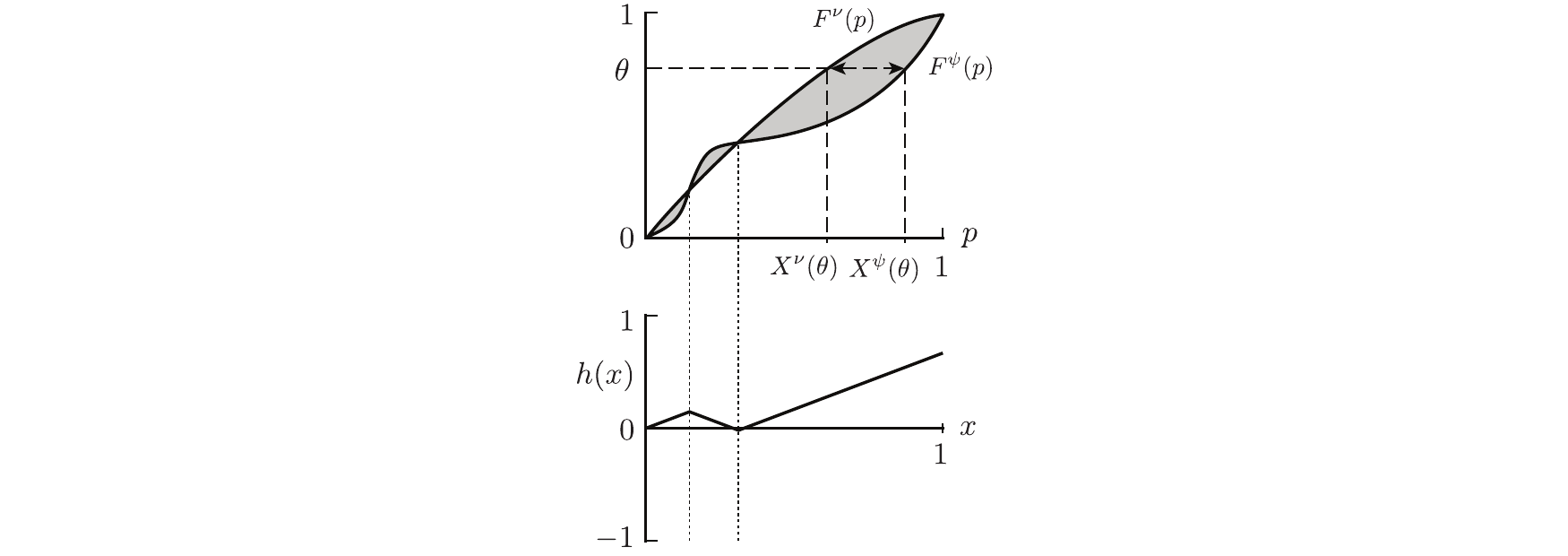}
  \caption{
(A) Illustration of the application of Laisant's formula to the expected difference of two random variables, see Equation~\eqref{eq:Laisant}. 
The expected difference of the two random variables $X^{\nu}$ and $X^{\psi}$ is obtained as the area enclosed by the two curves $X^{\nu}(\theta)$ and $X^{\psi}(\theta)$ for $\theta\in[0,1]$. Thus, the area is given by $\int_0^1\d \theta\  \left|X^{\nu}(\theta)-X^{\psi}(\theta)\right|$.  
On the other hand, the area between the two curves can be computed from the inverse functions to $X^{\nu}$ and $X^{\psi}$, which are the cumulative distribution functions $F^\nu(p) =  \int_{0}^p\d y\ \nu(y)$ and $F^\psi(p) =  \int_{0}^p\d y\ \psi(y)$, respectively, see Definition~\ref{def:rv}.
Therefore, the area is also given by $\int_0^1\d p\ \left|F^\nu(p)-F^\psi(p)\right|$. The rigorous argument follows with Laisant's formula.
(B) Definition of the function $h$ that is used to estimate $d(\nu,\psi)$. $h^\prime(x)=1$ if $F^\nu(p) > F^\psi(p)$ and $h^\prime(x)=  -1$ if $F^\nu(p) < F^\psi(p)$, and thus $\|h\|_L\leq 1$.
}
  \label{fig:Laisant}
\end{figure*}

Now, let $h:[0,1]\to\mathbb{R}$ be given by $h(x)\coloneqq\int_{0}^{x}\d p\ \text{sgn}\left(\int_0^p\d y\ \nu(y)-\psi(y)\right)$ (here sgn is the signum function), that is $h^\prime(x)=1$ if $\int_0^p \d y\ \nu(y) > \int_0^p\d y\ \psi(y)$ and $h^\prime(x)=  -1$ if $\int_0^p\d y\ \nu(y) < \int_0^p\d y\ \psi(y)$, see Figure~\ref{fig:Laisant}(B) for a sketch. 
In particular, it is $\|h\|_L\leq 1$. Therefore, one may use $h$ to estimate the supremum in the definition of the bounded Lipschitz metric $d(\cdot,\cdot)$ as:
\begin{align}\nonumber
d(\nu,\psi)\geq \int_0^1\d p\ h(p)(\nu(p)-\psi(p))\ .
\end{align}
Integration by parts yields:
\begin{align}\nonumber
d(\nu,\psi)&\geq\int_0^1\d p\  h^\prime(p)\int_0^p\d y\ (\nu(y)- \psi(y))\ ,\\\nonumber
&= \int_0^1\d p\ \left|\int_0^p\d y\ (\nu(y)- \psi(y))\right| \ .
\end{align}
Since both $\nu$ and $\psi$ are normalized to 1, the boundary terms vanish in the integration by parts above.

Together with Equation~\eqref{eq:Laisant} from above, one obtains the estimate of the lemma:
\begin{align}\nonumber
\mathbb{E}(|X^{\nu}-X^{\psi}|)\leq d(\nu,\psi) \ .
\end{align}
\end{proof}

\begin{lem}\label{lem:micro}
There exists a constant $0<Const<\infty$ such that for a given update step K one estimates:
\begin{align}
\left|\mathbb{E}\left(d(\rho_{N,K}^{(1)},\eta_{N,K}^\text{aux})\right)\right|\ \leq  e^{Const\cdot\frac{K}{N}}\left(\mathbb{E}\left( d(\rho_{N,0}^{(1)},\eta_{N,0}^\text{aux})\right) + \frac{K^{3/4}}{N}\right)\ .
\end{align}
\end{lem}

\begin{proof}
The proof of the lemma is based on a discrete Gr\"onwall's inequality. 

We first estimate how the distance between one realization of the microscopic process ($\rho_{N,k}^{(1)}$) and one realization of the auxiliary process ($\eta_{N,k}^\text{aux}$) propagates from at a certain update step $k$ to step $k+1$. 
This distance measures the error that occurs upon neglecting correlations of the individuals' production degrees. 
This error propagates on average from one update step $k$ to the next update step $k+1$ as follows:  
\begin{align}\label{eq:one_step}
\left|\mathbb{E}\left(d(\rho_{N,k+1}^{(1)},\eta_{N,k+1}^\text{aux})\right)-\mathbb{E}\left(d(\rho_{N,k}^{(1)},\eta_{N,k}^\text{aux})\right)\right|
\leq \frac{Const}{N}\mathbb{E}\left(\left|X^{\rho_{N,k}^{(1)}}-X^{\eta_{k}^\text{aux}}\right|\right) \ .
\end{align}
To see this estimate, we write:
\begin{align}\nonumber
d(&\rho_{N,k+1}^{(1)},\eta_{N,k+1}^\text{aux})\\ \nonumber
&=
 \sup_{\|f\|_L=1} \left| \int_0^1\d p\ f(p)\rho_{N,k+1}^{(1)}(p)-\int_0^1\d p\ f(p)\eta_{N,k+1}^\text{aux}(p) \right|\ ,\\\nonumber
&\leq\sup_{\|f\|_L=1} \left|  \int_0^1\d p\ f(p) \left(\rho_{N,k}^{(1)}(p)-\eta_{N,k}^\text{aux}(p)\right) \right|\\\nonumber 
&\quad + \frac{Const}{N}\sup_{\|f\|_L=1}  \int_0^1\d p\ f(p) \int_0^1\d \omega^{k+1}\ \big|\text{``realization (micro)''} - \text{``realization (aux)''}   \big|\ .
\end{align}
The latter estimate follows because the distance of the densities between consecutive update steps involves the change of at most four production degrees in the population and, thus, a change of probability mass of order $~\mathcal{O}(1/N)$ from $\rho_{N,k}^{(1)}$ to $\rho_{N,k+1}^{(1)}$ and from $\eta_{N,k}^\text{aux}$ to $\eta_{N,k+1}^\text{aux}$.

We further estimate:
\begin{align}\nonumber
&d(\rho_{N,k+1}^{(1)},\eta_{N,k+1}^\text{aux})\\ \nonumber
&\leq d(\rho_{N,k}^{(1)},\eta_{N,k}^\text{aux})\\\nonumber
&\quad +\frac{Const}{N}\sup_{\|f\|_L=1}  \int_0^1\d p\ f(p) \int_0^1\d \omega^{k+1}\ \left|X^{\rho_{N,k}^{(1)}}(\omega^{k+1}_1) - X^{\Phi(\eta^\text{aux}_{k+1})}(\omega^{k+1}_1)   \right|\\\nonumber
&\quad +\frac{Const}{N}\sup_{\|f\|_L=1}  \int_0^1\d p\ f(p) \int_0^1\d \omega^{k+1}\ \left|X^{\rho^{(1)}_{N,k}-\delta(p-p_{2k+1}^{\text{micro},-}(\omega^{k+1}_1))}(\omega^{k+1}_2) - X^{\eta_{k+1}^\text{aux}}(\omega^{k+1}_2)   \right|\\\nonumber
&\quad +\frac{Const}{N}\sup_{\|f\|_L=1}  \int_0^1\d p\ f(p) \int_0^1\d \omega^{k+1}\ \left|\omega_3^{k+1}\left(X^{\rho_{N,k}^{(1)}}(\omega^{k+1}_1) - X^{\Phi(\eta_{k+1}^\text{aux})}(\omega^{k+1}_1)\right)\right.\\\nonumber
&\qquad\qquad\qquad\qquad\qquad\qquad\qquad+\left.(1-\omega_3^{k+1})\left(R(\langle p\rangle_{\rho_{N,k}^{(1)}}) -R(\langle p\rangle_{\eta^\text{aux}_{k+1}})\right)   \right|\\\nonumber
&\quad +\frac{Const}{N}\sup_{\|f\|_L=1}  \int_0^1\d p\ f(p) \int_0^1\d \omega^{k+1}\ \left|\omega_4^{k+1}\left(X^{\rho_{N,k}^{(1)}}(\omega^{k+1}_1) - X^{\Phi(\eta_{k+1}^\text{aux})}(\omega^{k+1}_1)\right)\right.\\\nonumber
&\qquad\qquad\qquad\qquad\qquad\qquad\qquad+\left. (1-\omega_4^{k+1})\left(R(\langle p\rangle_{\rho_{N,k}^{(1)}}) -R(\langle p\rangle_{\eta^\text{aux}_{k+1}})\right)   \right|\\\nonumber
&\leq d(\rho_{N,k}^{(1)},\eta_{N,k}^\text{aux}) 
+ \frac{Const}{N}\mathbb{E}\left(\left|X^{\rho_{N,k}^{(1)}}-X^{\eta_{k+1}^\text{aux}}\right|\right)+\frac{Const}{N}\mathbb{E}\left(\left|R(\langle p \rangle_{\rho_{N,k}^{(1)}}) -R(\langle p\rangle_{\eta^\text{aux}_{k+1}})\right|\right)\ ,\\\nonumber
&\leq d(\rho_{N,k}^{(1)},\eta_{N,k}^\text{aux})
+ \frac{Const}{N}\mathbb{E}\left(\left|X^{\rho_{N,k}^{(1)}}-X^{\eta_{k}^\text{aux}}\right|\right)
\ .
\end{align}
The last estimate follows with the triangle inequality, the definition of the average auxiliary density~\eqref{def:aux}, and Lemma~\ref{lem:Xbyd}. 
By applying Lemma~\ref{lem:aux1} to the last line above, estimate~\eqref{eq:one_step} follows. 

We now determine how the growth of the average error from update step $k$ to $k+1$ depends upon the error at step~$k$.
By applying Lemma~\ref{lem:Xbyd} to the estimate~\eqref{eq:one_step}, which involves the average density of the auxiliary process and not the empirical density, one obtains (also note the different meanings of the expectation values taken above and below): 
\begin{align}\nonumber
&\left|\mathbb{E}\left(d(\rho_{N,k+1}^{(1)},\eta_{N,k+1}^\text{aux})\right)-\mathbb{E}\left(d(\rho_{N,k}^{(1)},\eta_{N,k}^\text{aux})\right)\right|\\\nonumber
&\quad\leq \frac{Const}{N}\mathbb{E}\left(d(\rho_{N,k}^{(1)},\eta_{k}^\text{aux})\right)\ , \\\nonumber
&\quad\leq \frac{Const}{N}\mathbb{E}\left(d(\rho_{N,k}^{(1)},\eta_{N,k}^\text{aux})\right)+\frac{Const}{N}\mathbb{E}\left(d(\eta_{N,k}^\text{aux},\eta_{k}^\text{aux})\right)\ ,\\\nonumber\label{eq:Gronwall_one_step}
&\quad\leq \frac{Const}{N}\mathbb{E}\left(d(\rho_{N,k}^{(1)},\eta_{N,k}^\text{aux})\right)+\frac{Const}{N}\left(\frac{k^{3/4}}{N}+\mathbb{E}\left(d(\eta^\text{aux}_{N,0},\eta^\text{aux}_{0})\right)\right)\ .
\end{align}
The estimate in the first line above follows with the triangle inequality, and the second estimate follows with the law of large numbers argument from Lemma~\ref{lem:LLN}. 
Essentially, the growth of the average error from update step $k$ to $k+1$ can be attributed to the following sources: (i) propagation of errors from update step $k$, (ii) creation of ``new'' errors at the update step $k+1$ because of the law of large numbers, and (iii) propagation of initial errors.

From the growth of errors between two consecutive update steps, the growth of errors for any given finite number of update steps $K$ can be controlled with Gr\"onwall's inequality as we show next.
Gr\"onwall's inequality for differentiable functions $u$ states that if $u^\prime (t)$ is bounded by $u^\prime (t)\leq \alpha u(t) +\alpha\beta$ with $\alpha, \beta \in \mathbb{R}$, then it follows that $u(t)$ is bounded by the solution of the differential equation given of the right-hand side ($u^\prime (t)= \alpha u(t) +\alpha\beta$) as $u(t)\leq u(0) e^{\alpha t}+ \beta(e^{\alpha t}-1)$.
In the spirit of a discrete version of Gr\"onwall's inequality applied to the Estimate~\eqref{eq:Gronwall_one_step}, one finds a constant $0<Const<\infty$ such that for a given update step $K$:
\begin{align}\nonumber
\left|\mathbb{E}\left(d(\rho_{N,K}^{(1)},\eta_{N,K}^\text{aux})\right)\right|\ \leq  e^{Const\cdot\frac{K}{N}}\left(\mathbb{E}\left(d(\rho_{N,0}^{(1)},\eta_{N,0}^\text{aux})\right)+\frac{K^{3/4}}{N}\right)\ ,
\end{align} 
which concludes the proof of Lemma~\ref{lem:micro}.
\end{proof}

\subsection{Convergence of the auxiliary to the macroscopic process ($\eta^\text{aux}_{\kappa(t)}\tos\rho_{t})$ -- Continuous time limit and control of time synchronization}

We now show that the mean-field density of the macroscopic process ($\rho_t$) converges in probability to the average density of the auxiliary process ($\eta^\text{aux}_{\kappa(t)}$) as $N\to \infty$, see Lemma~\ref{lem:macro}.
In other words, we show that the average auxiliary density at update step $\kappa(t)$ (the maximal number such that $\tau_{\kappa(t)}\leq t$) stays close to the mean-field density at the continuous time $t$. 

\begin{defn}\label{def:T}
Let $\nu\in(L^\infty)^\star$ be a time-dependent density function. We define the time evolution operator: 
\begin{align}
T_t(\nu)\coloneqq
\int_{0}^t\d {t^\prime}\ \big[2\lambda (1+s(b-c)\langle p\rangle_{\nu_{t^\prime}})
	\big(\delta(p-R(\langle p\rangle_{\nu_{t^\prime}}))-\nu(p,{t^\prime})\big)
+ (1-2\lambda)sc\big(\langle p\rangle_{\nu_{t^\prime}}-p\big)\nu(p,{t^\prime})\big] \ ,
\end{align}
with $\langle p\rangle_{\nu_{t^\prime}} = \int_0^1\d p\ p\, \nu(p,{t^\prime})$.
\end{defn}

With this definition, the time evolution of the empirical density of the macroscopic process is given by (see Definition~\eqref{eq:fitness} of the fitness and the mean-field equation~\eqref{eq:mean-field}): 
\begin{align}\label{eq:macro2b}
 \rho(p,t) = T_t(\rho)(p) +\rho_0(p)\ .
\end{align}

\begin{lem}\label{lem:Trho}
Let $\nu\in(L^\infty)^\star$ and $\psi\in(L^\infty)^\star$ be one-particle probability densities. 
Then 
\begin{align}\nonumber
d\left(T_t(\nu),T_t(\psi)\right)\leq Const\cdot\int_0^t\d {t^\prime}\ d(\nu_{t^\prime},\psi_{t^\prime})\ .
\end{align}
\end{lem}

\begin{proof}

After suitable rewriting, we use the triangle inequality to estimate $d\left(T_t(\nu),T_t(\psi)\right)$ as follows:
\begin{align}\nonumber
&d\left(T_t(\nu),T_t(\psi)\right)=\left\|T_t(\nu)-T_t(\psi)\right\|_{BL}\\\nonumber
&=\big\|\int_{0}^t \d{t^\prime}\ \big[2\lambda (1+s(b-c)\langle p\rangle_{\nu_{t^\prime}})\big(\delta(p-R(\langle p\rangle_{\nu_{t^\prime}}) -\nu(p,{t^\prime})\big)
+ (1-2\lambda)sc\big(\langle p\rangle_{\nu_{t^\prime}}-p\big)\nu(p,{t^\prime})\\\nonumber
&\quad-2\lambda (1+s(b-c)\langle p\rangle_{\psi_{t^\prime}})\big(\delta(p-R(\langle p\rangle_{\psi_{t^\prime}}) -\psi(p,{t^\prime})\big)
+ (1-2\lambda)sc\big(\langle p\rangle_{\psi_{t^\prime}}-p\big)\psi(p,{t^\prime}) \big] \big\|_{BL}\ ,\\\nonumber
&= \big\|\int_{0}^t \d{t^\prime}\ \big[2\lambda (1+s(b-c)\langle p\rangle_{\nu_{t^\prime}})\big(\delta(p-R(\langle p\rangle_{\nu_{t^\prime}}) -\nu(p,{t^\prime})\big)
+ (1-2\lambda)sc\big(\langle p\rangle_{\nu_{t^\prime}}-p\big)\nu(p,{t^\prime})\\\nonumber
&\quad-2\lambda (1+s(b-c)\langle p\rangle_{\psi_{t^\prime}})\big(\delta(p-R(\langle p\rangle_{\psi_{t^\prime}}) -\psi(p,{t^\prime})\big)
+ (1-2\lambda)sc\big(\langle p\rangle_{\psi_{t^\prime}}-p\big)\psi(p,{t^\prime}) \\\nonumber
&\quad+2\lambda (1+s(b-c)\langle p\rangle_{\nu_{t^\prime}})\big(\delta(p-R(\langle p\rangle_{\nu_{t^\prime}}) -\psi(p,{t^\prime})\big)
+ (1-2\lambda)sc\big(\langle p\rangle_{\nu_{t^\prime}}-p\big)\psi(p,{t^\prime})\\\nonumber
&\quad-2\lambda (1+s(b-c)\langle p\rangle_{\nu_{t^\prime}})\big(\delta(p-R(\langle p\rangle_{\nu_{t^\prime}}) -\psi(p,{t^\prime})\big)
+ (1-2\lambda)sc\big(\langle p\rangle_{\nu_{t^\prime}}-p\big)\psi(p,{t^\prime}) \big] \big\|_{BL}\ ,\\\nonumber
&\leq \big\|\int_{0}^t \d{t^\prime}\ \big[
2\lambda (1+s(b-c)\langle p\rangle_{\nu_{t^\prime}})\big(\psi(p,{t^\prime})-\nu(p,{t^\prime})\big)
+(1-2\lambda)sc\big(\langle p\rangle_{\nu_{t^\prime}}-p\big)\big(\nu(p,{t^\prime})-\psi(p,{t^\prime})\big)\big]\big\|_{BL}\\\nonumber
&\quad+\big\|\int_{0}^t \d{t^\prime}\ \big[
2\lambda (1+s(b-c)\langle p\rangle_{\nu_{t^\prime}})\delta(p-R(\langle p\rangle_{\nu_{t^\prime}})
-2\lambda (1+s(b-c)\langle p\rangle_{\psi_{t^\prime}})\delta(p-R(\langle p\rangle_{\psi_{t^\prime}}) \\\nonumber
&\quad +s(c-2\lambda b)\big(\langle p\rangle_{\nu_{t^\prime}}-\langle p\rangle_{\psi_{t^\prime}}\big)\psi(p,{t^\prime})\big]\big\|_{BL}\ .
\end{align}

Since $c,b,s,\lambda,\langle p\rangle_{\nu_{t^\prime}},\langle p\rangle_{\psi_{t^\prime}},\|\nu\|_{BL}$ and $\|\psi\|_{BL}$ are uniformly bounded, it follows that there exists a constant $0<Const<\infty$ such that
\begin{align}\nonumber
d\left(T_t(\nu),T_t(\psi)\right)
&\leq Const\cdot \int_{0}^t \d{t^\prime}\ \big[\big\| \psi(p,{t^\prime}) -\nu(p,{t^\prime})\big\|_{BL}+\left|\langle p\rangle_{\nu_{t^\prime}}-\langle p\rangle_{\psi_{t^\prime}}\right|\big]\ .
\end{align}
Using Lemma~\ref{lem:aux1} it follows that there exists (another) constant $0<Const<\infty$ such that
\begin{align}\nonumber
d\left(T_t(\nu),T_t(\psi)\right)
\leq Const\cdot\int_{0}^t \d{t^\prime}\ d(\nu_{t^\prime},\psi_{t^\prime}) \ .
\end{align}

\end{proof}

To prepare the continuous time limit, we show that the average time $\tau_k$, at which updates of the population occur, stays on average close to the continuous time. The proof proceeds by applying a law of large numbers argument.
Recall from Definition~\ref{def:microtime} that we denoted the random variable for the length of the time interval between the $k^{\text{th}}$ and $(k+1)^{\text{th}}$ update step as $\Delta\tau_k$ with $\langle\Delta\tau_k\rangle = N^{-1} \left(1+s(b-c)\langle p\rangle_k\right)^{-1}$. 
We also defined  $\tau_k= \sum_{l=1}^{k-1} \Delta\tau_l$ and  $\langle \tau_k\rangle =\sum_{l=1}^{k-1} \langle\Delta\tau_l \rangle$; $M(t)$ denotes the maximal natural number such that $\langle \tau_{M(t)}\rangle \leq t$ and $\kappa(t)$ is the random variable given by the maximal number such that $\tau_{\kappa(t)}\leq t$.

\begin{lem}\label{lem:time}
Let $t>0$. Then 
\begin{align}\nonumber
\mathbb{P}(|\kappa(t)-M(t)|\geq N^{3/4})\leq Const\cdot N^{-1/2}\ .
\end{align}
\end{lem}

\begin{proof}
The lemma is based on the law of large numbers.
Since the $\kappa(t)$ is monotonously increasing, it follows that 
\begin{align}\nonumber
\mathbb{P}(\kappa(t)<M(t)-N^{3/4})
&\leq \mathbb{P}(\tau_{M(t)-N^{3/4}}>t)\ ,\\\nonumber
&=\mathbb{P}(\tau_{M(t)-N^{3/4}}-\langle \tau_{M(T)-N^{3/4}}\rangle > t-\langle \tau_{M(t)-N^{3/4}}\rangle)\ ,\\\label{eq:probabsch}
&\leq \mathbb{P}\left(\left|\tau_{M(t)-N^{3/4}}-\langle \tau_{M(t)-N^{3/4}}\rangle\right|>\left|t-\langle \tau_{M(t)-N^{3/4}}\rangle\right|\right)\ .
\end{align}
Because the $\tau_k$ are independent of each other, it follows with Chebyshev's inequality that for any (possibly $N$-dependent) $\epsilon_N>0$:
\begin{align}\nonumber
\mathbb{P}\left(\left|\tau_{M(t)-N^{3/4}}-\langle \tau_{M(t)-N^{3/4}}\rangle\right|>\epsilon_N\right)
&\leq\epsilon_N^{-2}\text{Var}(\tau_{M(t)-N^{3/4}})\ ,\\\nonumber
&=\epsilon_N^{-2}\sum_{k=1}^{M(t)-N^{3/4}}\Var{\Delta\tau_k}\ ,\\\nonumber
&\leq Const\cdot\frac{N}{\epsilon_N^2 N^2}\ .
\end{align}
 Since the average lengths of time intervals between two update steps, $\langle\Delta\tau_k\rangle$, are bounded for all $k$ by some constant times $N^{-1}$, the respective variances are of order $N^{-2}$. 
The estimate in the last line above then follows by recalling that $M(t)\sim\mathcal{O}(N)$.

We choose $\epsilon_N \coloneqq \left|t-\langle \tau_{M(t)-N^{3/4}}\rangle\right|$ and estimate:
\begin{align}\nonumber
\epsilon_N 
&=  |\langle \tau_{M(t)-N^{3/4}}\rangle-t|\ ,\\\nonumber
&= \left|\sum_{k=1}^{M(t)-N^{3/4}}\langle\Delta\tau_k\rangle-t\right|\ ,\\\nonumber
&\leq  \left|\sum_{k=1}^{M(t)-N^{3/4}}\langle\Delta\tau_k\rangle-\sum_{k=1}^{M(t)}\langle\Delta\tau_k\rangle\right|+\frac{Const}{N}\ ,\\\nonumber
&\leq N^{3/4}\frac{Const}{N}\ ,\\\nonumber
&= Const\cdot N^{-1/4}\ .
\end{align}

Therefore, one obtains from Chebyshev's inequality with the chosen $\epsilon_N$:
\begin{align}\nonumber
\mathbb{P}\left(\left|\tau_{M(t)-N^{3/4}}-\langle \tau_{M(t)-N^{3/4}}\rangle\right|>\epsilon_N\right)
\leq Const\cdot N^{-1/2}\ .
\end{align}
From Equation~\eqref{eq:probabsch} one obtains:
\begin{align}\nonumber
\mathbb{P}(\kappa(t)<M(t)-N^{3/4})\leq CN^{-1/2}\ .
\end{align}
In the same way one shows that:
\begin{align}\nonumber
\mathbb{P}(\kappa(t)>M(t)+N^{3/4})\leq CN^{-1/2}\ ,
\end{align}
and the lemma follows.
\end{proof}

After these preparatory steps, we now proceed with the following lemma, which estimates the average distance between the mean-field density at the real time $t$ and the average auxiliary density at update step $\kappa(t)$ (that is the random variable given by the maximal number such that $\tau_{\kappa(t)}\leq t$). 
The proof exploits Gr\"onwall's inequality.
%
\begin{lem}\label{lem:macro}

For any $t>0$, one estimates: 
\begin{align}\nonumber
\mathbb{E}\left(d(\eta^\text{aux}_{\kappa(t)},\rho_{t}) \right)\leq Const(t)\cdot N^{-1/4}\ .	
\end{align}
%
\end{lem}
%

\begin{proof}
We apply the triangle inequality and estimate:
\begin{align}\nonumber
\mathbb{E}\left(d(\eta^\text{aux}_{\kappa(t)},\rho_{t})\right)
&\leq  \mathbb{E}\left(d(\eta^\text{aux}_{\kappa(t)},\eta^\text{aux}_{M(t)})\right)
+ d(\eta^\text{aux}_{M(t)},\rho_{t})
\ .	
\end{align}
Note that the expectation values above are taken with respect to sampling the update times.
The first summand (i) addresses the distance of the average auxiliary density between different update steps; namely between a single realization of update steps ($\kappa(t)$) up to the given time $t$ and the average number of update steps ($M(t)$) up to time $t$.
By a law of large numbers argument, we show below that:
\begin{align}\label{eq:timing_time}
 \mathbb{E}\left(d(\eta^\text{aux}_{\kappa(t)},\eta^\text{aux}_{M(t)})\right)\leq Const\cdot N^{-1/4}
\ ,
\end{align}
The second summand (ii) governs the distance of the auxiliary process at average times to the macroscopic process at the real time $t$. We show below that the propagation of errors due to different timings of the auxiliary process and the macroscopic process are controlled by applying Gr\"onwall's inequality, and estimate
\begin{align}\label{eq:timing_groenwall}
d(\eta^\text{aux}_{M(t)},\rho_{t})\leq \frac{1}{N}e^{Const\cdot t}\  .	
\end{align}

(i) First, we estimate $ \mathbb{E}\left(d(\eta^\text{aux}_{\kappa(t)},\eta^\text{aux}_{M(t)})\right)$ in Equation~\eqref{eq:timing_time} by splitting up the expectation value as follows:
\begin{align}\nonumber
 \mathbb{E}\left(d(\eta^\text{aux}_{\kappa(t)},\eta^\text{aux}_{M(t)})\right)
&\leq \sup\left\{\|\eta^\text{aux}_{\kappa(t)}-\eta^\text{aux}_{M(t)}\|_{BL}: |\kappa(t)-M(t)|\geq N^{3/4}\right\}\cdot\mathbb{P}(|\kappa(t)-M(t)|\geq N^{3/4}) \\\nonumber
&\quad+ \sup\left\{\|\eta^\text{aux}_{\kappa(t)}-\eta^\text{aux}_{M(t)}\|_{BL}: |\kappa(t)-M(t)|\leq N^{3/4}\right\}\cdot\mathbb{P}(|\kappa(t)-M(t)|\leq N^{3/4})\ ,\\\nonumber
&\leq \sup\left\{\|\eta^\text{aux}_{\kappa(t)}-\eta^\text{aux}_{M(t)}\|_{BL}\right\}\mathbb{P}(|\kappa(t)-M(t)|\geq N^{3/4}) \\\nonumber
&\quad+\sup\left\{\|\eta^\text{aux}_{\kappa(t)}-\eta^\text{aux}_{M(t)}\|_{BL}: |\kappa(t)-M(t)|\leq N^{3/4}\right\}
\ .
\end{align}
Using the fact that for any probability densities $\nu$ and $\psi$ it is $d_{BL}(\nu,\psi)\leq\|\nu\|_1+\|\psi\|_1=2$, one estimates: 
\begin{align}\nonumber
\sup\left\{\|\eta^\text{aux}_{\kappa(t)}-\eta^\text{aux}_{M(t)}\|_{BL}\right\}\leq 2\ .
\end{align}
From Lemma~\ref{lem:time}, we obtain 
\begin{align}\nonumber
\mathbb{P}(|\kappa(t)-M(t)|\geq N^{3/4})\leq Const \cdot N^{-1/2}\ .
\end{align}
Since $\|\eta_k^\text{aux}\|_{BL}$ and $\|\Phi(\eta_k^\text{aux})\|_{BL}$ are bounded, it follows that 
\begin{align}\nonumber
\|\eta_{k+1}^\text{aux}-\eta_{k}^\text{aux}\|_{BL}
= N^{-1}\|-\eta_{k}^\text{aux}+\Phi(\eta_{k}^\text{aux})\|_{BL}
\leq Const/N
\ .
\end{align}
Therefore, one obtains:
\begin{align}\nonumber
\sup\left\{\|\eta^\text{aux}_{\kappa(t)}-\eta^\text{aux}_{M(t)}\|_{BL}: |\kappa(t)-M(t)|\leq N^{3/4}\right\}\leq Const\cdot N^{3/4}N^{-1} = Const\cdot N^{-1/4}
\ ,
\end{align}
and the estimate in Equation~\eqref{eq:timing_time} follows as
\begin{align}\nonumber
 \mathbb{E}\left(d(\eta^\text{aux}_{\kappa(t)},\eta^\text{aux}_{M(t)})\right)\leq Const\cdot N^{-1/4}\ .
\end{align}

(ii) Second, we show the estimate in Equation~\eqref{eq:timing_groenwall} for $d(\eta^\text{aux}_{M(t)},\rho_{t})$. Recall that by Definition~\ref{def:Phi}, we have
\begin{align}\nonumber
\Phi(\eta^\text{aux}_k)(p)=
2\lambda \delta\left(p-R(\langle p\rangle_{\eta^\text{aux}_k})\right)
+    (1-2\lambda)\frac{1+s\big(b\langle p\rangle_{\eta^\text{aux}_k}-cp\big)}{1+s(b-c)\langle p\rangle_{\eta^\text{aux}_k}}\eta^\text{aux}_k(p)\ . 
\end{align}
We write the average density of the auxiliary process at update step $k$ by applying Definition~\ref{def:aux} iteratively:
\begin{align}\nonumber
\eta^\text{aux}_{k}(p)&=\eta^\text{aux}_{k-1}(p)-\frac{1}{N}\eta^\text{aux}_{k-1}(p)+ \frac{1}{N}\Phi(\eta^\text{aux}_{k-1})(p)\ , \\\nonumber
&=\eta^\text{aux}_{0}(p)+\frac{1}{N}\sum_{j=0}^{k-1} \big[\Phi(\eta^\text{aux}_j)(p)-\eta^\text{aux}_{j}(p)\big]\ ,\\\nonumber
&=\eta^\text{aux}_{0}(p)+\sum_{j=0}^{k-1} \left[\frac{2\lambda}{N} \delta\left(p-R(\langle p\rangle_{\eta^\text{aux}_j})\right)
+    \frac{1-2\lambda}{N}\frac{1+s\big(b\langle p\rangle_{\eta^\text{aux}_j}-cp\big)}{1+s(b-c)\langle p\rangle_{\eta^\text{aux}_j}}\eta^\text{aux}_j(p)-\frac{1}{N}\eta^\text{aux}_{j}(p)\right]\ ,\\\nonumber
&=\eta^\text{aux}_{0}(p)+\sum_{j=0}^{k-1} N^{-1} \left(1+s(b-c)\langle p\rangle_{\eta^\text{aux}_{j}}\right)^{-1}\big[ 2\lambda\left(1+s(b-c)\langle p\rangle_{\eta^\text{aux}_{j}}\right)\left(\delta(p-R(\langle p\rangle_{\eta^\text{aux}_j})- \eta^\text{aux}_j\right)\\\nonumber
&\quad\quad\quad\quad\quad\quad\quad\quad\quad\quad+     (1-2\lambda)sc\big(\langle p\rangle_{\eta^\text{aux}_j}-p\big)\eta^\text{aux}_j(p)\big]\ ,\\\nonumber
&=\eta^\text{aux}_{0}(p)+\int_0^{\langle \tau_k \rangle}\d{t^\prime}\ 
\big[2\lambda(1+s(b-c)\langle p\rangle_{\eta^\text{aux}_{M({t^\prime})}})\left( \delta\big(p-R(\langle p\rangle_{\eta^\text{aux}_{M({t^\prime})}})\big)-\eta^\text{aux}_{M({t^\prime})}(p)\right)\\\nonumber
&\quad\quad\quad\quad\quad\quad\quad\quad\quad\quad+   (1-2\lambda)sc\big(\langle p\rangle_{\eta^\text{aux}_{M({t^\prime})}}-p)\big)\eta^\text{aux}_{M({t^\prime})}(p)\big]\ ,
\end{align}
where, in the last line, it was exploited that the $j^\text{th}$ update step occurs after an average time $\langle \Delta\tau_j \rangle = N^{-1} \left(1+s(b-c)\langle p\rangle_{\eta^\text{aux}_{j}}\right)^{-1}$.
Furthermore, the $L^1$-norm of the above integrand is bounded. Thus, one finds a constant $0<Const<\infty$ such that:
\begin{align}\nonumber
\|\eta^\text{aux}_{M(t)}&-T_t(\eta^\text{aux}_{M(t)})-\eta^\text{aux}_{0}\|_{BL}\ ,\\\nonumber
&\leq\big\|\int_{\langle \tau_k \rangle}^{t}\d {t^\prime}\ 
\big[2\lambda(1+s(b-c)\langle p\rangle_{\eta^\text{aux}_{M({t^\prime})}})\left( \delta\big(p-R(\langle p\rangle_{\eta^\text{aux}_{M({t^\prime})}})\big)-\eta^\text{aux}_{M({t^\prime})}(p)\right)\ ,\\\nonumber
&\quad\quad\quad\quad\quad\quad\quad\quad\quad\quad +(1-2\lambda)sc\big(\langle p\rangle_{\eta^\text{aux}_{M({t^\prime})}}-p)\big)\eta^\text{aux}_{M({t^\prime})}(p)\big]
\big\|_{BL}\ ,\\\nonumber
&\leq Const\cdot (t-\langle \tau_k \rangle)\ ,\\\nonumber
&\leq  Const/N\ .
\end{align}
Therefore, one estimates with the triangle inequality:
\begin{align}\nonumber
\|\eta^\text{aux}_{M(t)}-\rho_{t}\|_{BL}&=\|\eta^\text{aux}_{M(t)}-T_t(\eta^\text{aux}_{M(t)})-\eta^\text{aux}_{0}+T_t(\eta^\text{aux}_{M(t)})+\eta^\text{aux}_{0}-\rho_{t}\|_{BL}\ ,\\\nonumber
&\leq \|\eta^\text{aux}_{M(t)}-T_t(\eta^\text{aux}_{M(t)})-\eta^\text{aux}_{0}\|_{BL}
+\|T_t(\eta^\text{aux}_{M(t)})+\eta^\text{aux}_{0}-T_t(\rho) -\rho_0\|_{BL}\ ,\\\nonumber
&\leq Const/N + \|T_t(\eta^\text{aux}_{M(t)})-T_t(\rho)\|_{BL} + d(\eta^\text{aux}_{0},\rho_0)\ ,\\\nonumber
&\leq Const/N+ Const\cdot\int_{0}^t \d{t^\prime}\ \big\| \eta^\text{aux}_{t^\prime} -\rho_{t^\prime}\big\|_{BL}\ .
\end{align}
The last line follows with Lemma~\ref{lem:Trho} and with $d(\eta^\text{aux}_{0},\rho_0) = 0$ because $\eta^\text{aux}_{0}=\rho_0$. 
By applying Gr\"onwall's inequality to the last line above, one obtains the estimate in Equation~\eqref{eq:timing_groenwall}:
\begin{align}\nonumber
d(\eta^\text{aux}_{M(t)},\rho_{t})\leq \frac{1}{N}e^{Const\cdot t}\ .
\end{align}

Combining the estimates for summand (i) in Equation~\eqref{eq:timing_time} and summand (ii) in Equation~\eqref{eq:timing_groenwall}, Lemma~\ref{lem:macro} follows.
\end{proof}

\subsection{Proof of the theorem}

\begin{proof}

We estimate the expectation value of the distance between the empirical one-particle density of the microscopic stochastic many-particle process and the mean-field density by the estimates obtained in Lemmas~\ref{lem:LLN},~\ref{lem:micro} and~\ref{lem:macro}. 
 In total, one finds a $0<Const<\infty$ such that 
\begin{align}\nonumber
\mathbb{E}\left(d(\rho_N^{(1)}(t),\rho(t))\right)
&\leq 
\mathbb{E}\left(d(\rho_{N,\kappa(t)}^{(1)},\eta_{N,\kappa(t)}^\text{aux})\right)
+\mathbb{E}\left(d(\eta_{N,\kappa(t)}^\text{aux},\eta_{\kappa(t)}^\text{aux})\right)
+\mathbb{E}\left(d(\eta_{\kappa(t)}^\text{aux}, \rho_t)\right)\ ,\\\nonumber
&\leq  e^{Const\cdot\frac{M(t)}{N}}\left(\mathbb{E}\left( d(\rho_{N,0}^{(1)},\eta_{N,0}^\text{aux})\right) 
+ \frac{M(t)^{3/4}}{N}\right)\\\nonumber
&\qquad +Const\cdot  \frac{M(t)^{3/4}}{N}+\mathbb{E}\left(d(\eta^\text{aux}_{N,0},\eta^\text{aux}_{0})\right)
+ Const(t)\cdot N^{-1/4}\ .
\end{align}
Note that  $\eta^\text{aux}_{0} =\rho_0$. 
Note also that expectation values above are taken with respect to both  sampling production degrees and sampling update times.
Furthermore, for a given time $t$, $M(t)$ is of order $N$. 
Thus, one estimates:
\begin{align}\nonumber
\mathbb{E}\left(d(\rho_N^{(1)}(t),\rho(t))\right)\ 
 \leq Const(t)\cdot
 \left(
 \mathbb{E}\left(d(\eta^\text{aux}_{N,0},\rho_{0})
 \right) 
 +  N^{-1/4}
 \right)\ ,
\end{align}
with some constant $0<Const<\infty$ that depends on the chosen time $t$.
Applying Markov's inequality establishes an error estimate of the convergence with an (even $N$-dependent) $\epsilon_N$:
\begin{align}\nonumber
\mathbb{P}\left(d(\rho_N^{(1)}(t),\rho(t))>\epsilon_N\right)
&\leq \frac{\mathbb{E}\left(d(\rho_N^{(1)}(t),\rho(t))\right)}{\epsilon_N} \\\nonumber
&\leq \frac{Const(t)}{\epsilon_N}\cdot
\left(\mathbb{E}\left(d(\eta^\text{aux}_{N,0},\rho_{0})\right) +  N^{-1/4}\right)\ ,
\end{align} 
which proves Theorem \ref{thm:main}.

\end{proof}

\subsection*{Acknowledgements}
This research was supported by the German Excellence Initiative via the program ``Nanosystems Initiative Munich'' (NIM) and by the Deutsche Forschungsgemeinschaft within the framework SPP1617 on ``Phenotypic heterogeneity and sociobiology of bacterial populations'' (through grant FR 850/11-1, 2). 
We thank Mauro Mobilia and Stefano Duca for discussions about applications of our work to the field of opinion dynamics.
The authors declare no conflict of interest.

\clearpage


\begin{thebibliography}{10}
\providecommand{\url}[1]{{#1}}
\providecommand{\urlprefix}{URL }
\expandafter\ifx\csname urlstyle\endcsname\relax
  \providecommand{\doi}[1]{DOI~\discretionary{}{}{}#1}\else
  \providecommand{\doi}{DOI~\discretionary{}{}{}\begingroup
  \urlstyle{rm}\Url}\fi

\bibitem{Ackermann2015}
Ackermann, M.: A functional perspective on phenotypic heterogeneity in
  microorganisms.
\newblock Nat Rev Micro \textbf{13}(8), 497--508 (2015).
\newblock \urlprefix\url{10.1038/nrmicro3491}

\bibitem{Balian2007}
Balian, R.: From Microphysics to Macrophysics, vol.~2, 2nd edn.
\newblock Springer-Verlag Berlin (2007)

\bibitem{Bauer2017}
Bauer, M., Knebel, J., Lechner, M., Pickl, P., Frey, E.: Ecological feedback in
  quorum-sensing microbial populations can induce heterogeneous production of
  autoinducers.
\newblock eLife \textbf{6}, e25,773 (2017).
\newblock \urlprefix\url{10.7554/eLife.25773}

\bibitem{BenNaim2003}
Ben-Naim, E., Krapivsky, P., Redner, S.: Bifurcations and patterns in
  compromise processes.
\newblock Physica D: Nonlinear Phenomena \textbf{183}(3), 190 -- 204 (2003).
\newblock \urlprefix\url{10.1016/S0167-2789(03)00171-4}

\bibitem{Blythe2007b}
Blythe, R.A., Evans, M.R.: Nonequilibrium steady states of matrix product form:
  A solver's guide.
\newblock J. Phys. A-Math. Theor. \textbf{40}(46), R333 (2007).
\newblock \urlprefix\url{10.1088/1751-8113/40/46/R01}

\bibitem{Blythe2007}
Blythe, R.A., McKane, A.J.: Stochastic models of evolution in genetics, ecology
  and linguistics.
\newblock Journal of Statistical Mechanics: Theory and Experiment
  \textbf{2007}(07), P07,018 (2007).
\newblock \urlprefix\url{10.1088/1742-5468/2007/07/P07018}



\bibitem{Boers2015}
Boers, N., Pickl, P.: On mean field limits for dynamical systems. \newblock J. Stat. Phys., \textbf{164}(1), 1–16, (2015).

\bibitem{Bomze1990}
Bomze, I.M.: Dynamical aspects of evolutionary stability.
\newblock Monatshefte f{\"u}r Mathematik \textbf{110}(3), 189--206 (1990).
\newblock \urlprefix\url{10.1007/BF01301675}

\bibitem{Braun1977}
Braun, W., Hepp, K.: The vlasov dynamics and its fluctuations in the 1/n limit
  of interacting classical particles.
\newblock Communications in Mathematical Physics \textbf{56}(2), 101--113
  (1977).
\newblock \urlprefix\url{10.1007/BF01611497}

\bibitem{Canizares2017}
Canizares, A., Pickl, P.: Microscopic derivation of the keller-segel equation
  in the sub-critical regime.
\newblock arXiv:1703.04376  (2017).
\newblock \urlprefix\url{arXiv:1703.04376}

\bibitem{Carlen2013}
Carlen, E., Chatelin, R., Degond, P., Wennberg, B.: Kinetic hierarchy and
  propagation of chaos in biological swarm models.
\newblock Physica D: Nonlinear Phenomena \textbf{260}, 90 -- 111 (2013).
\newblock \urlprefix\url{10.1016/j.physd.2012.05.013}.
\newblock Emergent Behaviour in Multi-particle Systems with Non-local
  Interactions

\bibitem{Cattiaux2016}
Cattiaux, P., P{\'e}d{\`e}ches, L.: The 2-d stochastic keller-segel particle
  model: existence and uniqueness.
\newblock arXiv:1601.08026  (2016).
\newblock \urlprefix\url{arXiv:1601.08026}

\bibitem{Chou2011}
Chou, T., Mallick, K., Zia, R.K.P.: Non-equilibrium statistical mechanics: from
  a paradigmatic model to biological transport.
\newblock Reports on Progress in Physics \textbf{74}(11), 116,601 (2011).
\newblock \urlprefix\url{http://stacks.iop.org/0034-4885/74/i=11/a=116601}

\bibitem{Cressman2005}
Cressman, R.: Stability of the replicator equation with continuous strategy
  space.
\newblock Mathematical Social Sciences \textbf{50}(2), 127 -- 147 (2005).
\newblock \urlprefix\url{10.1016/j.mathsocsci.2005.03.001}

\bibitem{Derrida1992}
Derrida, B., Domany, E., Mukamel, D.: An exact solution of a one-dimensional
  asymmetric exclusion model with open boundaries.
\newblock Journal of Statistical Physics \textbf{69}(3), 667--687 (1992).
\newblock \urlprefix\url{10.1007/BF01050430}

\bibitem{Diggle2007}
Diggle, S.P., Griffin, A.S., Campbell, G.S., West, S.A.: Cooperation and
  conflict in quorum-sensing bacterial populations.
\newblock Nature \textbf{450}(7168), 411--414 (2007).
\newblock \urlprefix\url{10.1038/nature06279}

\bibitem{Drees2014}
Drees, B., Reiger, M., Jung, K., Bischofs, I.B.: A modular view of the
  diversity of cell-density-encoding schemes in bacterial quorum-sensing
  systems.
\newblock Biophysical Journal \textbf{107}(1), 266--277 (2014).
\newblock \urlprefix\url{10.1016/j.bpj.2014.05.031}

\bibitem{Edelstein1988}
Edelstein-Keshet, L.: Mathematical models in biology, 1st edn.
\newblock Random House / Birkh{\"a}user mathematics series (1988)

\bibitem{Figalli2017}
Figalli, A., Kang, M.J.: A rigorous derivation from the kinetic cucker-smale
  model to the pressureless euler system with nonlocal alignment.
\newblock arXiv:1702.08087  (2016).
\newblock \urlprefix\url{arXiv:1702.08087}

\bibitem{Gardiner}
Gardiner, C.: {Stochastic Methods: A Handbook for the Natural and Social
  Sciences}.
\newblock Springer, Berlin (2009)

\bibitem{Garmyn2011}
Garmyn, D., Gal, L., Briandet, R., Guilbaud, M., Lema{\^\i}tre, J.P., Hartmann,
  A., Piveteau, P.: Evidence of autoinduction heterogeneity via expression of
  the agr system of listeria monocytogenes at the single-cell level.
\newblock Applied and Environmental Microbiology \textbf{77}(17), 6286--6289
  (2011).
\newblock \urlprefix\url{10.1128/AEM.02891-10}

\bibitem{deGennes1995}
de~Gennes P~G, J, P.: The Physics of Liquid Crystals, 2nd edn.
\newblock Oxford University Press (1995)

\bibitem{Gillespie1976}
Gillespie, D.T.: A general method for numerically simulating the stochastic
  time evolution of coupled chemical reactions.
\newblock J. Comput. Phys. \textbf{22}(4), 403--434 (1976).
\newblock \urlprefix\url{10.1016/0021-9991(76)90041-3}

\bibitem{Gillespie1977}
Gillespie, D.T.: Exact stochastic simulation of coupled chemical reactions.
\newblock J. Phys. Chem. \textbf{81}(25), 2340--2361 (1977).
\newblock \urlprefix\url{10.1021/j100540a008}

\bibitem{Graner1992}
Graner, F., Glazier, J.A.: Simulation of biological cell sorting using a
  two-dimensional extended potts model.
\newblock Phys. Rev. Lett. \textbf{69}, 2013--2016 (1992).
\newblock \urlprefix\url{10.1103/PhysRevLett.69.2013}

\bibitem{Grote2015}
Grote, J., Krysciak, D., Streit, W.R.: Phenotypic heterogeneity, a phenomenon
  that may explain why quorum sensing not always results in truly homogenous
  cell behavior.
\newblock Applied and Environmental Microbiology  (2015).
\newblock \urlprefix\url{10.1128/AEM.00900-15}

\bibitem{He2003}
He, X., Chang, W., Pierce, D.L., Seib, L.O., Wagner, J., Fuqua, C.: Quorum
  sensing in \textit{Rhizobium sp.} strain ngr234 regulates conjugal transfer
  (tra) gene expression and influences growth rate.
\newblock Journal of Bacteriology \textbf{185}(3), 809--822 (2003).
\newblock \urlprefix\url{10.1128/JB.185.3.809-822.2003}

\bibitem{Hense2015}
Hense, B.A., Schuster, M.: Core principles of bacterial autoinducer systems.
\newblock Microbiology and Molecular Biology Reviews \textbf{79}(1), 153--169
  (2015).
\newblock \urlprefix\url{10.1128/MMBR.00024-14}

\bibitem{Hofbauer2003}
Hofbauer, J., Sigmund, K.: Evolutionary game dynamics.
\newblock Bull. Amer. Math. Soc. \textbf{40}(4), 479--519 (2003).
\newblock \urlprefix\url{10.1090/S0273-0979-03-00988-1}

\bibitem{Huang2017}
Huang, H., Liu, J.G.: Error estimate of a random particle blob method for the
  keller-segel equation.
\newblock Mathematics of Computation \textbf{86}(308), 2719--2744 (2017).
\newblock \urlprefix\url{10.1090/mcom/3174}

\bibitem{Kadar2011}
Kadar, M.: {Statistical Physics of Particles}, 3 edn.
\newblock Cambridge University Press, Cambridge, UK (2007)

\bibitem{Klumpp2003}
Klumpp, S., Lipowsky, R.: Traffic of molecular motors through tube-like
  compartments.
\newblock Journal of Statistical Physics \textbf{113}(1), 233--268 (2003).
\newblock \urlprefix\url{10.1023/A:1025778922620}

\bibitem{Krug1991}
Krug, J.: Boundary-induced phase transitions in driven diffusive systems.
\newblock Phys. Rev. Lett. \textbf{67}, 1882--1885 (1991).
\newblock \urlprefix\url{10.1103/PhysRevLett.67.1882}

\bibitem{Lanford1975}
Lanford, O.E.: Time evolution of large classical systems, pp. 1--111.
\newblock Springer Berlin Heidelberg, Berlin, Heidelberg (1975).
\newblock \urlprefix\url{10.1007/3-540-07171-7_1}


\bibitem{Liu2017} Huang, H., Liu, J.G.: Error estimate of a random particle blob method for the Keller-Segel equation Math. Comp. \textbf{86}, 2719-2744  (2017). 

\bibitem{Lorenz2007}
Lorenz, J.: Continuous opinion dynamics under bounded confidence: A survey.
\newblock International Journal of Modern Physics C \textbf{18}(12), 1819--1838
  (2007).
\newblock \urlprefix\url{10.1142/S0129183107011789}

\bibitem{Maire2015}
Maire, T., Youk, H.: Molecular-level tuning of cellular autonomy controls the
  collective behaviors of cell populations.
\newblock Cell Systems \textbf{1}(5), 349--360 (2015).
\newblock \urlprefix\url{10.1016/j.cels.2015.10.012}

\bibitem{Marchetti2013}
Marchetti, M.C., Joanny, J.F., Ramaswamy, S., Liverpool, T.B., Prost, J., Rao,
  M., Simha, R.A.: Hydrodynamics of soft active matter.
\newblock Rev. Mod. Phys. \textbf{85}, 1143--1189 (2013).
\newblock \urlprefix\url{10.1103/RevModPhys.85.1143}

\bibitem{McGill2007}
McGill, B.J., Brown, J.S.: Evolutionary game theory and adaptive dynamics of
  continuous traits.
\newblock Annual Review of Ecology, Evolution, and Systematics \textbf{38},
  403--435 (2007).
\newblock \urlprefix\url{10.1146/annurev.ecolsys.36.091704.175}

\bibitem{Murray2002}
Murray, J.D.: Mathematical Biology. I: An Introduction, 3rd edn.
\newblock Springer-Verlag New York (2002)

\bibitem{Neunzert1974}
Neunzert, H., Wick, J.: Die approximation der l{\"o}sung von
  integro-differentialgleichungen durch endliche punktmengen.
\newblock In: R.~Ansorge, W.~T{\"o}rnig (eds.) Numerische Behandlung
  nichtlinearer Integrodifferential-und Differentialgleichungen, pp. 275--290.
  Springer Berlin Heidelberg, Berlin, Heidelberg (1974).
\newblock \urlprefix\url{10.1007/BFb0060678}

\bibitem{Nowak2004}
Nowak, M.A., Sigmund, K.: Evolutionary dynamics of biological games.
\newblock Science \textbf{303}(5659), 793--799 (2004).
\newblock \urlprefix\url{10.1126/science.1093411}

\bibitem{Oechssler2001}
Oechssler, J., Riedel, F.: Evolutionary dynamics on infinite strategy spaces.
\newblock Economic Theory \textbf{17}(1), 141--162 (2001).
\newblock \urlprefix\url{10.1007/PL00004092}

\bibitem{Parmeggiani2003}
Parmeggiani, A., Franosch, T., Frey, E.: {Phase coexistence in driven
  one-dimensional transport}.
\newblock Phys. Rev. Lett. \textbf{90}(February), 086,601 (2003).
\newblock \urlprefix\url{10.1103/PhysRevLett.90.086601}

\bibitem{Ruparell2016}
Ruparell, A., Dubern, J.F., Ortori, C.A., Harrison, F., Halliday, N.M., Emtage,
  A., Ashawesh, M.M., Laughton, C.A., Diggle, S.P., Williams, P., Barrett,
  D.A., Hardie, K.R.: The fitness burden imposed by synthesising quorum sensing
  signals.
\newblock Scientific Reports \textbf{6}, 33,101 (2016).
\newblock \urlprefix\url{10.1038/srep33101}

\bibitem{Schaller2010}
Schaller, V., Weber, C., Semmrich, C., Frey, E., Bausch, A.R.: Polar patterns of driven filaments.
\newblock Nature \textbf{467}, 73 EP (2010).
\newblock \urlprefix\url{10.1038/nature09312}

\bibitem{Segerer2015}
Segerer, F.J., Th\"uroff, F., Piera~Alberola, A., Frey, E., R\"adler, J.O.:
  Emergence and persistence of collective cell migration on small circular
  micropatterns.
\newblock Phys. Rev. Lett. \textbf{114}, 228,102 (2015).
\newblock \urlprefix\url{10.1103/PhysRevLett.114.228102}

\bibitem{Sepulveda2013}
Sep{\'u}lveda, N., Petitjean, L., Cochet, O., Grasland-Mongrain, E., Silberzan,
  P., Hakim, V.: Collective cell motion in an epithelial sheet can be
  quantitatively described by a stochastic interacting particle model.
\newblock PLOS Computational Biology \textbf{9}(3), 1--12 (2013).
\newblock \urlprefix\url{10.1371/journal.pcbi.1002944}

\bibitem{Spitzer1970}
Spitzer, F.: Interaction of markov processes.
\newblock Advances in Mathematics \textbf{5}(2), 246 -- 290 (1970).
\newblock \urlprefix\url{10.1016/0001-8708(70)90034-4}

\bibitem{Spohn1991}
Spohn, H.: Large Scale Dynamics of Interacting Particles.
\newblock Springer, Heidelberg (1991)

\bibitem{VanKampen2007}
{Van Kampen}, N.G.: {Stochastic Process in Physics and Chemistry}.
\newblock Elsevier, Amsterdam (2007)

\bibitem{Vicsek1995}
Vicsek, T., Czir\'ok, A., Ben-Jacob, E., Cohen, I., Shochet, O.: Novel type of
  phase transition in a system of self-driven particles.
\newblock Phys. Rev. Lett. \textbf{75}, 1226--1229 (1995).
\newblock \urlprefix\url{10.1103/PhysRevLett.75.1226}

\bibitem{Weber2017}
Weber, M.F., Frey, E.: Master equations and the theory of stochastic path
  integrals.
\newblock Reports on Progress in Physics \textbf{80}(4), 046,601 (2017).
\newblock \urlprefix\url{10.1088/1361-6633/aa5ae2}

\bibitem{Williams2008}
Williams, J.W., Cui, X., Levchenko, A., Stevens, A.M.: Robust and sensitive
  control of a quorum-sensing circuit by two interlocked feedback loops.
\newblock Molecular Systems Biology \textbf{4}(1) (2008).
\newblock \urlprefix\url{10.1038/msb.2008.70}

\bibitem{Yin2017}
Yin Q., C.L.G.S.: The mean field kinetic equation for a pedestrian flow model:
  Existence and uniqueness of weak solution.
\newblock arXiv:1709.02686  (2017).
\newblock \urlprefix\url{arXiv:1709.02686}

\bibitem{Zwanzig2001}
Zwanzig, R.: Nonequilibrium statistical mechanics.
\newblock Oxford University Press (2001)

\end{thebibliography}


\end{document}